\documentclass[%
 reprint,
 amsmath,amssymb,
 aps
]{revtex4-2}
\usepackage{hyperref}
\usepackage[capitalise,noabbrev]{cleveref}
\usepackage{amsfonts,amssymb,amsthm,amsmath}
\usepackage{graphicx}
\usepackage{dcolumn}
\usepackage{bm}
\usepackage[mathlines]{lineno}
\usepackage{color}
\usepackage[capitalise]{cleveref}
\usepackage{tabularray}
\usepackage{booktabs}
\newtheorem{theorem}{Theorem}
\newtheorem{remark}{Remark}
\newtheorem{example}{Example}
\newtheorem{lemma}{Lemma}

\newtheorem{proposition}{Proposition}
\newtheorem{definition}{Definition}
\usepackage{mathtools}
\usepackage{comment}
\usepackage{soul}

\DeclarePairedDelimiterX\braket[2]{\langle}{\rangle}{#1\,\delimsize\vert\,\mathopen{}#2}

\newcommand*{\spec}{{\mathrm{Sp}}}

\newcommand*{\dist}{{\mathrm{dist}}}
\newcommand*{\diag}{{\mathrm{diag}}} 
\newcommand*{\dH}{{d_{\mathrm{H}}}}

\usepackage{enumitem}
\usepackage{tabularx}

\newcommand*{\dd}{{\,\mathrm{d}}}

\newcommand*{\MH}{{\mathcal{M}_{\mathrm{H}}}}

\newcommand*{\OS}{{\Omega_{h}}}

\newcommand*{\OT}{{\Omega_{td}}}

\begin{document}

\preprint{APS/123-QED}

\title{Computation and Verification of Spectra for Non-Hermitian Systems}

\author{Catherine Drysdale}
 \altaffiliation{MARS, Lancaster University}
\email{c.drysdale@lancaster.ac.uk}
\author{Matthew Colbrook}
 \altaffiliation{DAMTP, University of Cambridge}
\email{mjc249@cam.ac.uk}
\author{Michael T. M. Woodley}
\altaffiliation{Department of Physics, University of Bath}
\email{mw2970@bath.ac.uk}

\date{\today}

\begin{abstract}
We establish a connection between quantum mechanics and computation, revealing fundamental limitations for algorithms computing spectra, especially in non-Hermitian settings. Introducing the concept of locally trivial pseudospectra (\textit{LTP}), we show such assumptions are necessary for spectral computation. \textit{LTP} adapts dynamically to system energies, enabling spectral analysis across a broad class of challenging non-Hermitian problems. Exploiting this framework, we overcome a longstanding obstacle by computing the eigenvalues and eigenfunctions of the imaginary cubic oscillator $H_{\mathrm{B}} = p^2 + i x^3$ with error bounds and no spurious modes—yielding, to our knowledge, the first such error-controlled result. We confirm, for instance, the 100th eigenvalue as $627.6947122484365113526737029011536\ldots$. Here, truncation-induced $\mathcal{PT}$-symmetry breaking causes spurious eigenvalues—a pitfall our method avoids, highlighting the link between truncation and physics. Finally, we illustrate the approach's generality via spectral computations for a range of physically relevant operators. This letter provides a rigorous framework linking computational theory to quantum mechanics and offers a precise tool for spectral calculations with error bounds.
\vspace{-5mm}   
\end{abstract}

\maketitle

Among the axioms of quantum mechanics \cite{dirac1981principles,von2018mathematical}, Hermiticity of Hamiltonians stands out as mathematical rather than physical. While this condition ensures that eigenvalues of observables are real, it is sufficient rather than necessary. In \cite{BenderRealSpectrainNon-HermitianHamiltonians}, Bender \& Boettcher analyzed eigenvalues of the imaginary cubic oscillator:
\begin{equation}\setlength\abovedisplayskip{4pt}\setlength\belowdisplayskip{4pt}
\label{eqn:H}
    H_{\mathrm{B}} = p^2 + i x^3.
\end{equation}
Remarkably, $H_{\mathrm{B}}$ is not Hermitian, yet has a real, discrete, and bounded-below spectrum \cite{dorey2001spectral,caliceti1980perturbation}. This Hamiltonian is also connected to fascinating physical properties, via a more general Lagrangian formulation in scalar quantum field theory, such as $\mathcal{L} = \frac{1}{2}(\partial\phi)^{2} + \frac{1}{2}m^{2}\phi^{2} - g(i\phi)^{N}$ $(N\geq2)$ \cite{Bender1999}, which can allow for asymptotic freedom where conventional theories lack it \cite{Bender2000,Bender2001}. Implications for supersymmetry \cite{Bender1998} and the Higgs boson \cite{Bender2000} have been investigated. Moreover, $\mathcal{PT}$-symmetric quantum electrodynamics has been introduced as an analog of the imaginary cubic field theory ($N=3$) \cite{Bender2005}, following earlier work on calculating the fine structure constant \cite{Bender1999}.

Research on non-Hermitian quantum systems has flourished, supported by a robust complex extension of quantum mechanics \cite{BenderComplexExtension,BenderNonHermitianHamiltoniansSense,bender2019pt}. There have been extensive experimental studies across electromagnetism \cite{PhysRevLett.103.093902,ruter2010observation,feng2011nonreciprocal,regensburger2012parity,xiao2021observation,peng2014parity,bittner2012pt,zhang2020synthetic,weimann2017topologically}, acoustics \cite{shi2016accessing,auregan2017pt}, electronics \cite{chtchelkatchev2012stimulation,schindler2011experimental,bender2013observation,cao2022fully,yang2022observation}, mechanical systems \cite{bender2013observation2}, metamaterial design \cite{li2024experimental}, bound states in the continuum and quantum scattering \cite{soley2023experimentally}, solitons in wave optics \cite{musslimani2008optical} and transitions in optical structures \cite{zhang2016observation,chong2011pt,xiao2017observation}. The intersection of non-Hermiticity with topology is a major source of recent interest, and encompasses phenomena that include edge states \cite{Banerjee2023}, defect states \cite{Zeng2023}, exceptional rings \cite{Bergholtz2021}, superconductivity \cite{Ohnmacht2025}, and novel singularities \cite{Hu2024}, in addition to providing new system classification schemes \cite{Yang2024,Liu2019}.

Eigenvalues are crucial in physical applications of non-Hermitian systems. Systems associated with gain (due to energy sources) and loss (due to the environment) have coupled eigenvalue pairs \cite{PhysRevLett.103.093902,el2018non}, where a perturbation to one results in a change in the other, allowing robustness. This is particularly useful in laser technology \cite{feng2014single,hodaei2014parity}, where optical loss can lead to poor beam quality, and in wireless power transfer \cite{assawaworrarit2017robust}, where it is important to maintain transfer efficiency. Additionally, non-Hermitian systems can exhibit exceptional points, with applications in sensing \cite{chen2016pt,liu2016metrology}.

Eigenvalues also underpin the foundations of $\mathcal{PT}$-symmetric quantum mechanics. A Hamiltonian $H$ is $\mathcal{PT}$-symmetric if it is invariant under the combined parity ($p\mapsto -p,x \mapsto -x$) and time-reversal ($p\mapsto -p,x \mapsto x,i\mapsto-i$) operators. If $H$ has unbroken $\mathcal{PT}$ symmetry, it has $\mathcal{CPT}$ symmetry for a hidden symmetry $\mathcal{C}$ \cite{CarlMBenderHiddenSymmetry, BenderHamiltonianHermitian}. The dynamics are unitary with respect to a $\mathcal{CPT}$ inner product $\langle f,g\rangle_{\mathcal{CPT}}=\langle \mathcal{CPT}f,g\rangle$. The development of this theory is analogous
to the problem that Dirac encountered in formulating the spinor wave equation in relativistic quantum theory. However, one must first find the eigenstates of $H$ to define $\mathcal{C}$ and the physical Hilbert space \cite{BenderCubicInteraction}. Hence, $\mathcal{PT}$-symmetric quantum mechanics is a bootstrap theory: $\mathcal{C}$ is spectrally determined and accurately calculating eigenvalues is essential. 

Computing eigenvalues of non-Hermitian operators is notoriously difficult. Previous approaches for $H_{\mathrm{B}}$, such as direct integration \cite{BenderRealSpectrainNon-HermitianHamiltonians}, domain truncation \cite{BoegliSieglTretter}, variational saddle-points \cite{bender1999variational}, perturbation theory \cite{BenderCubicInteraction,Mostafazadeh_2006}, and coupled moment problems \cite{handy2003moment}, lack explicit error bounds. Truncation to finite-dimensional problems often introduces spurious eigenvalues (see \cref{fig:pseudospectra}). Addressing this remains a key open problem in non-Hermitian quantum mechanics.

Recent work in computability theory has focused on determining whether algorithms can compute spectra for broad classes of operators. This has culminated in the solvability complexity index (SCI) hierarchy \cite{Hansen_JAMS,ben2015can,colbrook2020PhD}, which classifies algorithmic limits for spectral computations. Notably, recent efforts have targeted non-Hermitian spectral problems \cite{chandler2024spectral}.

This letter addresses the challenges of computing spectra for a broad class of non-Hermitian operators, focusing on $H_{\mathrm{B}}$, due to its central role in $\mathcal{PT}$ symmetry and its striking physical implications, particularly in quantum field theory. To overcome these challenges, we introduce the concept of ``locally trivial pseudospectra" (\textit{LTP}). For $\mathcal{PT}$-symmetric operators, \textit{LTP} exploits the local dependence of $\langle \cdot,\cdot\rangle_{\mathcal{CPT}}$ on the spectral data of $H$, essentially building local approximations of $\mathcal{C}$. This corresponds to Hilbert spaces that adapt dynamically to observed energy. Moreover, \textit{LTP} extends beyond $\mathcal{PT}$-symmetric operators---see \cref{fig:other_examples} for further examples.

We first prove that computing spectra of non-Hermitian operators is generally impossible without two key physical assumptions. The first concerns the interactions of an operator representation, ensuring that $H\psi$ can be computed with error bounds for wavefunctions $\psi$. For $H_{\mathrm{B}}$, we achieve this via rectangular truncations of a finite-range representation and explain how to extend the approach to long-range interactions. The second assumption, encoded by \textit{LTP}, provides local control over the operator’s non-Hermiticity, physically related to the $\mathcal{C}$ operator in the case of $\mathcal{PT}$-symmetric systems. Precisely, it bounds the growth of the resolvent norm near the spectrum (we define \textit{LTP} after some background). By combining control over interaction terms with \textit{LTP}, we bound the action of $\mathcal{C}$ on subspaces. In this way, \textit{LTP} connects physics with computation: the ability to compute spectra hinges on foundational quantum-mechanical structure. We then perform the first verified (i.e., with error bounds) computations of eigenvalues and eigenfunctions of $H_{\mathrm{B}}$, using a diffusion semigroup $\exp(-tH_{\mathrm{B}})$ and asymptotics to bound spectral projections and demonstrate \textit{LTP}. This dynamic perspective generalizes to a broad class of non-Hermitian operators.

This letter lays the foundation for rigorously verified spectra across a broad class of systems by introducing a new technology in non-Hermitian quantum mechanics. It provides a powerful tool for accurately computing and verifying spectral properties, bridging the gap between physical and computational feasibility. Beyond resolving the spectrum of $H_{\mathrm{B}}$, this work has broad implications in physics and applications of spectral computations (further illustrated in the End Matter, e.g., for exceptional points), and has foundational implications in unifying quantum mechanics and computation.

\textit{Fundamental limitations of computing spectra.---}
Two major, and separate, challenges in computing a Hamiltonian's spectrum are long-range interactions and non-Hermiticity. To illustrate this, we use the SCI hierarchy and adopt a general framework for defining computational problems---applicable to all such problems found in physics---and a general algorithm for solving them (see S.M.). A computational problem is defined by a function representing a physical quantity (e.g., spectra), with convergence to a solution measured within a metric space. One must also specify the information available to the algorithm, such as interaction data. This setup allows us to determine whether a convergent sequence of algorithms exists to solve the problem.

\begin{figure}[t]
\centering
\includegraphics[width=0.5\textwidth,trim=1mm 1mm 1mm 1mm,clip]{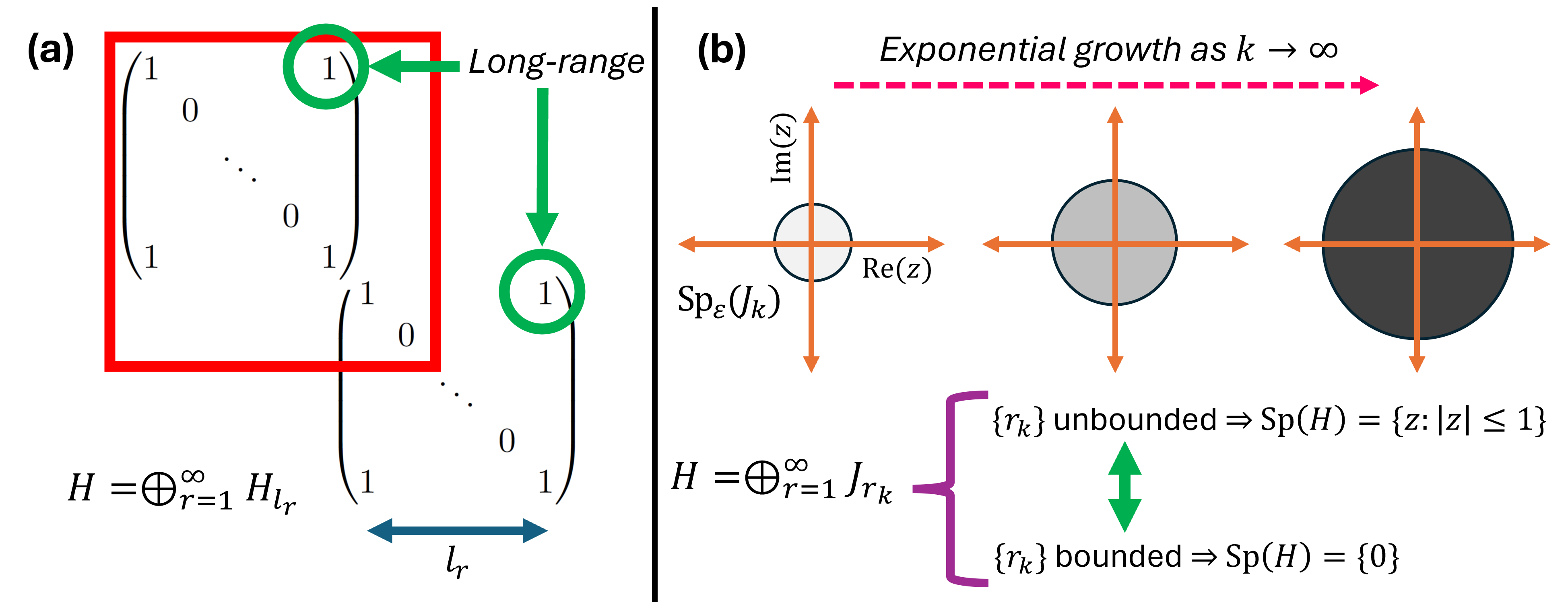}\vspace{-2mm}
\caption{\textbf{(a)}: Illustration of statement (A). We construct direct sums of Hermitian matrices with eigenvalues $\{0,2\}$ and long-range terms. Block sizes are chosen to `trick' the sequence of algorithms into approximating an eigenvalue  $1$. The red box shows the information read by the algorithm when making the approximation. \textbf{(b)}: Illustration of statement (B). We use direct sums of Jordan matrices whose pseudospectra grow rapidly with size (top panel). By varying block sizes, the algorithm is forced to oscillate (green arrow) between $\{0\}$ (when sizes are bounded) and the unit disk (otherwise).\vspace{-2mm}}
\label{fig:proof_idea}
\end{figure}

Consider a Hamiltonian $H$ with matrix elements $\{h_{ij}\}_{i,j\in\mathbb{N}}$.
We seek a sequence of algorithms $\{\Gamma_{m}\}$ that, given $\{h_{ij}\}$ as input, produce outputs $\Gamma_{m}(H)$ converging to the spectrum $\mathrm{Sp}(H)$ as $m \to \infty$.

First, let $H$ be Hermitian ($h_{ij} = \overline{h_{ji}}$), composed as a direct sum of smaller Hermitian blocks that may include long-range terms \cite{Lakkaraju2022, Eldredge2017, Schauss2012}—see Fig.~\ref{fig:proof_idea}(a). We prove (S.M.) that (A) no algorithmic sequence can converge to $\mathrm{Sp}(H)$ for all such inputs. The core obstacle, labeled \textbf{(R1)}, is the inability to control the presence of long-range terms.

Second, let $H$ be a (possibly non-Hermitian) tridiagonal Hamiltonian ($h_{ij}=0\text{ if }|i-j|>1$), formed as a direct sum of Jordan blocks. We prove (S.M.) that (B) no algorithmic sequence can converge to $\mathrm{Sp}(H)$ for all such inputs. The key issue, labeled \textbf{(R2)}, is the uncontrolled blow-up of $\|(H-zI)^{-1}\|$ as $z\rightarrow\mathrm{Sp}(H)$, a phenomenon unique to non-Hermitian operators.

These computational boundaries are fundamental and hold for \textit{any} algorithm, regardless of the operations used. They also extend to the computation of eigenvectors. However, results (A) and (B) do not imply that computing $\mathrm{Sp}(H)$ for a specific operator is impossible, only that any method must address \textbf{(R1)} and \textbf{(R2)}.

\textit{Introducing LTP to tackle \textbf{(R2)}.---}Non-Hermitian properties of a Hamiltonian $H$ can be studied through pseudospectra \cite{trefethen2005spectra}:
$$
\mathrm{Sp}_\epsilon(H):=\{z\in\mathbb{C}:\|(H-zI)^{-1}\|^{-1}\leq\epsilon\}, \quad\epsilon>0,
$$
where $\smash{\|(H-zI)^{-1}\|^{-1}=0}$ for $z\in\mathrm{Sp}(H)$. Equivalently, $\mathrm{Sp}_\epsilon(H)$ is the union of $\mathrm{Sp}(H+H')$ for perturbations $\|H'\|\leq\epsilon$ \cite[Prop. 4.15]{roch1996c}. Hence, $\mathrm{Sp}_\epsilon(H)$ tells us how far an $\epsilon$-sized perturbation can perturb $\mathrm{Sp}(H)$.

To address \textbf{(R2)}, we use $\|(H-zI)^{-1}\|^{-1}$ to approximate $\mathrm{dist}(z, \mathrm{Sp}(H))$, the distance from $z$ to the spectrum. For general $H$, we have $\|(H - zI)^{-1}\|^{-1} \leq \mathrm{dist}(z, \mathrm{Sp}(H))$, with equality when $H$ is Hermitian \cite{davies2000pseudospectra}. For non-Hermitian $H$, this quantity can significantly underestimate the true distance. If $H$ is quasi-Hermitian—i.e., $H^* \Theta = \Theta H$ for a bounded, positive Hermitian $\Theta$ with bounded inverse—then $\mathrm{dist}(z, \mathrm{Sp}(H)) \leq \sqrt{\|\Theta\|\|\Theta^{-1}\|} \|(H - zI)^{-1}\|^{-1}$, a condition known as trivial pseudospectrum \cite{KSTV}. In unbroken $\mathcal{PT}$-symmetric cases, this condition corresponds to a bounded transformation between $\langle \cdot,\cdot \rangle$ and $\langle \cdot,\cdot\rangle_{\mathcal{CPT}}$.  However, many important non-Hermitian operators, such as $H_\mathrm{B}$, do not exhibit trivial pseudospectra \cite{PhysRevD.86.121702}.

\begin{figure}[t]
\centering
\includegraphics[width=0.48\textwidth]{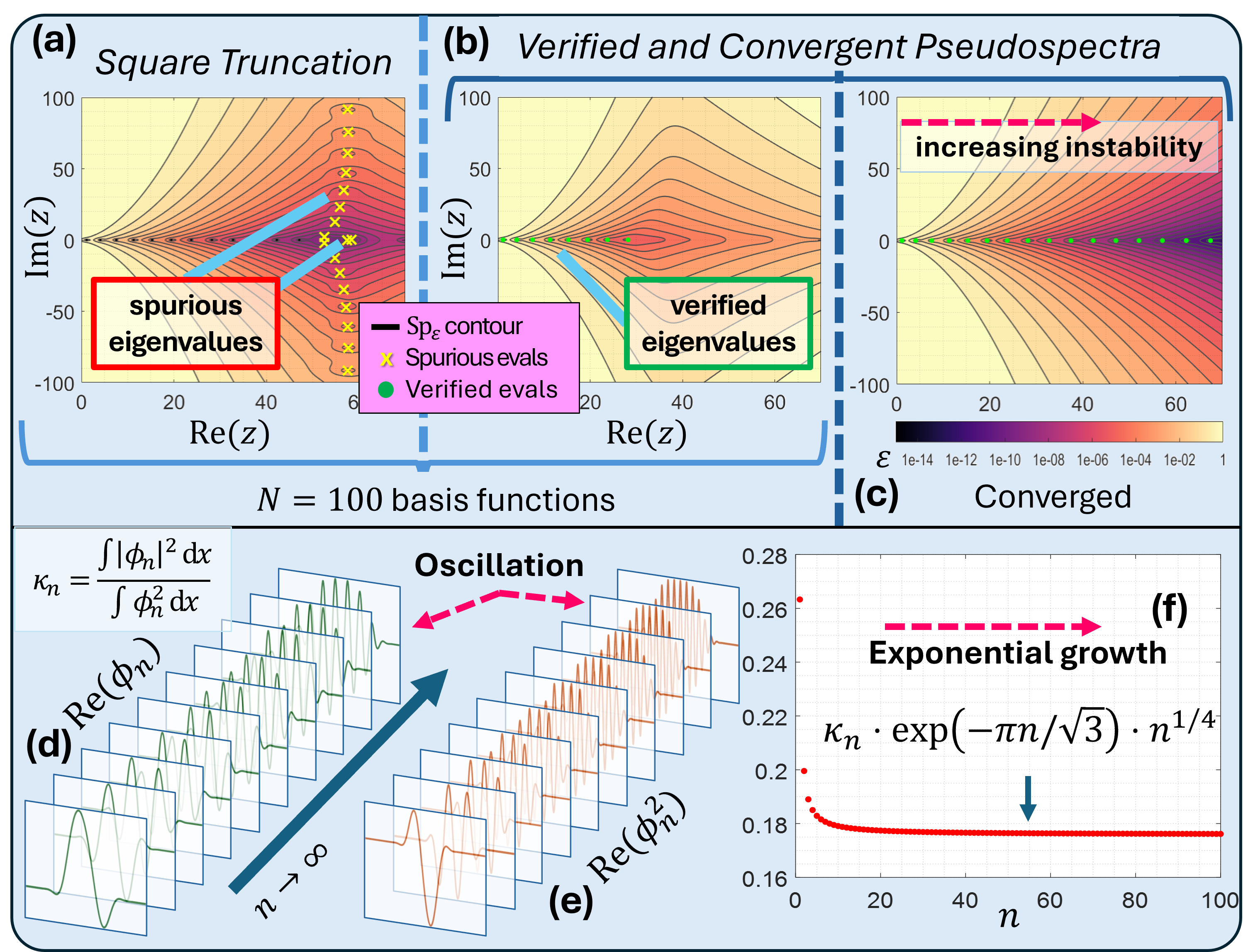}\vspace{-2mm}
\caption{\textbf{(a)}: Pseudospectra approximations via truncated Hermite expansions (square truncations), which break $\mathcal{PT}$-symmetry and produce spurious eigenvalues (yellow crosses). \textbf{(b,c)}: Rectangular truncations preserve $\mathcal{PT}$-symmetry, yield convergent approximations, and remain within the true pseudospectrum. Color indicates $\epsilon=\|(H_{\mathrm{B}}-zI)^{-1}\|^{-1}$. \textbf{(d,e)}: Eigenfunction oscillations. \textbf{(f)}: Exponential instability of eigenvalues, with $\lim_{n\rightarrow\infty}\kappa_n\exp(-n\pi/\sqrt{3})n^{1/4}\approx 0.176$.\vspace{-2mm}}\label{fig:pseudospectra}
\end{figure}

We say an operator $H$ has \textit{LTP} if, for every compact set $K\subset\mathbb{C}$, there exists $C_K>0$ such that for all $\epsilon>0$:
\begin{equation}\label{eq:locally_trivial_pseudopsectrum}
\setlength\abovedisplayskip{2pt}\setlength\belowdisplayskip{2pt}
\mathrm{Sp}_{\epsilon}(H)\cap K\subset\{z\in\mathbb{C}:\mathrm{dist}(z,\mathrm{Sp}(H))\leq C_K\epsilon\}.
\end{equation}
This means the pseudospectrum (and thus eigenvalue perturbations) is locally controlled by the distance to the spectrum. For example, if $H$ has a compact resolvent and all eigenvalues are semisimple, then $\|(H-zI)^{-1}\| = \mathcal{O}(1/|z-\lambda|)$ near each eigenvalue $\lambda$ \cite[Sec. III.6.5]{kato2013perturbation}. Since each compact set contains only finitely many such $\lambda$, $H$ satisfies \textit{LTP}. The concept also extends to operators with continuous spectra, generalizing the stable pseudospectral behavior of Hermitian operators.

For an unbroken $\mathcal{PT}$-symmetric operator $H$, \textit{LTP} reflects a locally bounded transformation between $\langle \cdot,\cdot \rangle$ and $\langle \cdot,\cdot\rangle_{\mathcal{CPT}}$. To see this, consider the Riesz projection corresponding to the $n$th eigenvalue of $H$, given by
$$
\mathcal{Q}_n=\frac{-1}{2\pi i }\int_{S_n} (H-zI)^{-1}\,\mathrm{d} z.
$$
Here, the contour $S_n$ wraps once around $\lambda_n$ and no other eigenvalues \cite[page 178]{kato2013perturbation}. Since each $\lambda_n$ is simple, the condition number of the eigenvalue $\lambda_n$ with eigenfunction $\phi_n$ is (simplifying with $\mathcal{PT}$-symmetry) \cite{aslanyan2000spectral}
\begin{align}\setlength\abovedisplayskip{4pt}\setlength\belowdisplayskip{4pt}
&\kappa_n=\|\mathcal{Q}_n\|=\frac{\langle \phi_n,\phi_n\rangle}{|\langle \mathcal{CPT}\phi_n,\phi_n\rangle|}=\frac{\langle \phi_n,\phi_n\rangle}{\langle \phi_n,\phi_n\rangle_{\mathcal{CPT}}},\text{ and}\label{eq:cond_def}\\
&\|(H-zI)^{-1}\|={\kappa_n}/{\dist(z,\lambda_n)}+\mathcal{O}(1),\text{ as }z\rightarrow \lambda_n,\label{eq:cond_meaning}
\end{align}
which relates $\kappa_n$ with the local behavior of $\mathrm{Sp}_\epsilon(H)$.

To compute $\mathrm{Sp}(H_\mathrm{B})$, we will determine constants $C_K$ in Eq. (\ref{eq:locally_trivial_pseudopsectrum}) for $H = H_{\mathrm{B}}$ and suitable sets $K$ that cover $\mathbb{C}$. Eigenvalues of $H_{\mathrm{B}}$ satisfy \cite{BenderRealSpectrainNon-HermitianHamiltonians}
$$\setlength\abovedisplayskip{4pt}\setlength\belowdisplayskip{4pt}
\lambda_n\!=\!\left[\frac{2\Gamma\left(\frac{11}{6}\right)(n-1/2)\sqrt{\pi}}{\sqrt{3}\Gamma\left(\frac{4}{3}\right)}\right]^{\frac{6}{5}}\!\!\!+\mathcal{O}\left(n^{-4/5}\right),\!\quad n=1,2,\ldots.
$$
There exists a constant $C_{\mathrm{B}}$ with \cite{henry2014spectralcubic}
\begin{equation}
\label{HB_asympt}\setlength\abovedisplayskip{4pt}\setlength\belowdisplayskip{4pt}
\lim_{n\rightarrow\infty}\kappa_n\exp(-n\pi/\sqrt{3})n^{1/4}=C_{\mathrm{B}}.
\end{equation}
Hence, $\lambda_n$ become exponentially unstable as $n \to \infty$.

\vspace{1mm}
\textit{Rectangular truncations to tackle \textbf{(R1)}.}---
We represent $H_{\mathrm{B}}$ using Hermite functions
$
u_m(x)=e^{-x^2/2}H_{m}(x),
$
for $m\in\mathbb{Z}_{\geq 0}$,
where
$
H_m(x)
$
is the $m$th Hermite polynomial. This yields a banded infinite matrix. Panel (a) of \cref{fig:pseudospectra} shows pseudospectra of a naïve $N\times N$ square truncation, which overestimates pseudospectra and produces spurious eigenvalues, due to neglected interactions with the full system---breaking $\mathcal{PT}$ symmetry.

\begin{figure}
\centering
\includegraphics[width=0.48\textwidth,trim=8mm 4mm 4mm 4mm,clip]{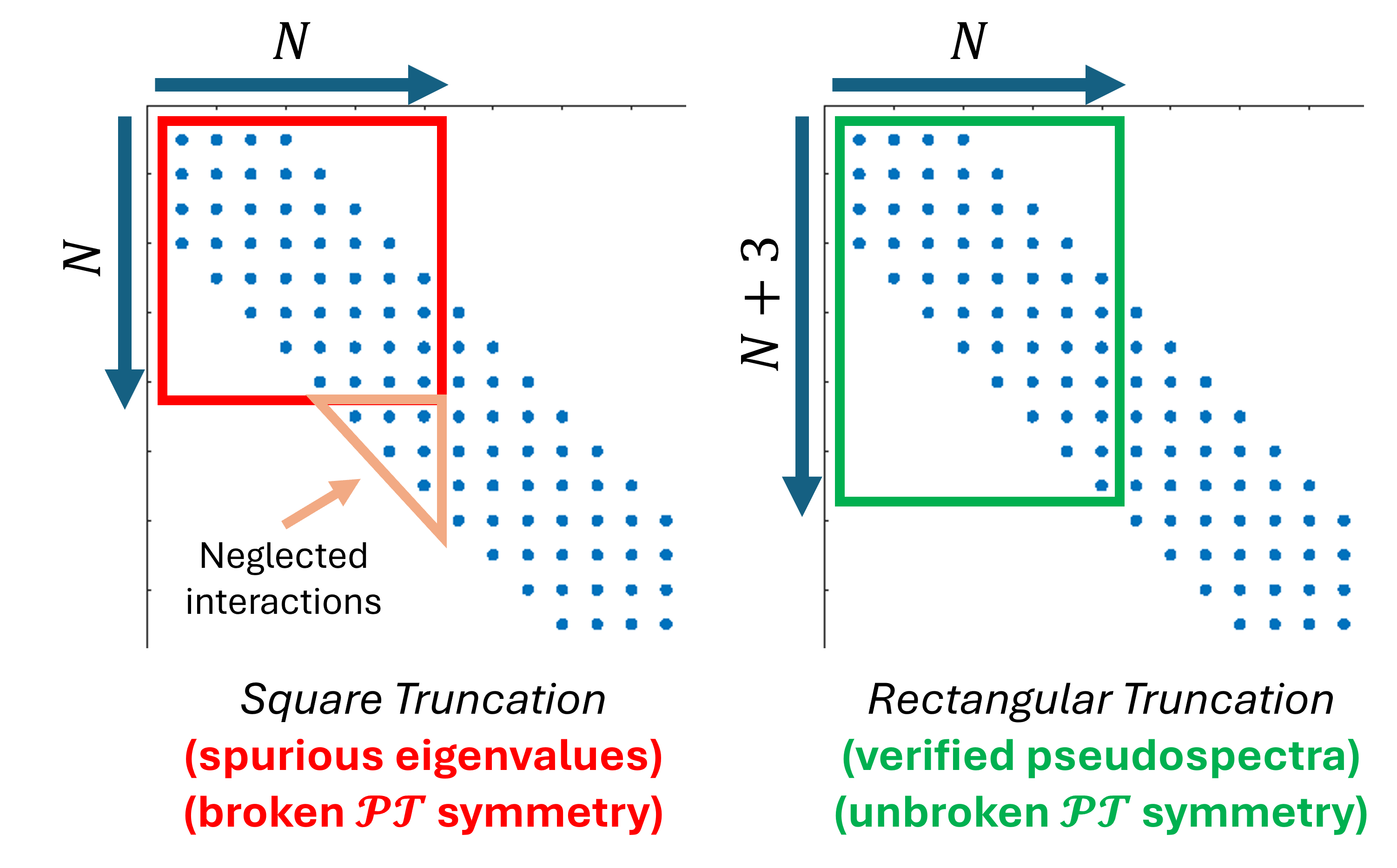}\vspace{-2mm}
\caption{Left: Square truncations of $H_{\mathrm{B}}$'s infinite matrix (blue dots indicate nonzero entries) omit important, physically relevant interactions. A direct check shows that the truncations have broken $\mathcal{PT}$-symmetry. Right: Rectangular truncations capture these terms, enabling verified spectra and pseudospectra. Properties of $H_\mathrm{B}$ (including unbroken $\mathcal{PT}$-symmetry) are preserved under rectangular truncations when we restrict to the subspace (range of $\mathcal{P}_N$).\vspace{-1mm}}
\label{fig:rect_truncation}
\end{figure}

We instead use \textit{rectangular} truncations \cite{PhysRevLett.122.250201} (see \cref{fig:rect_truncation}) to capture the missing interactions, preserve $\mathcal{PT}$-symmetry, overcome \textbf{(R1)}, and enable verification. Let $\mathcal{P}_N$ be the projection onto $\mathrm{span}\{u_0,\ldots,u_{N-1}\}$. Due to the banded structure of $H_{\mathrm{B}}$, we have $\mathrm{Sp}_\epsilon(H_{\mathrm{B}}\mathcal{P}_N)=\mathrm{Sp}_\epsilon(\mathcal{P}_{N+3}H_{\mathrm{B}}\mathcal{P}_N)$, where the latter is the pseudospectrum of a finite $(N+3)\times N$ matrix that can be computed. This approximation satisfies the bound $\mathrm{Sp}_\epsilon(H_{\mathrm{B}}\mathcal{P}_N)\subset\mathrm{Sp}_{\epsilon}(H_{\mathrm{B}})$ and converges as $N\rightarrow\infty$. Panels (b) and (c) of \cref{fig:pseudospectra} illustrate this convergence.

Crucially, rectangular truncations enable approximation of $\kappa_n$ in Eq. \eqref{eq:cond_def}. For $z \approx \lambda_n$, the right-singular vector of $\mathcal{P}_{N+3}(H_{\mathrm{B}}-zI)\mathcal{P}_N$ (associated with the smallest singular value) approximates $\phi_n$. We compute an initial approximation $\smash{\hat\phi_n\approx\phi_n}$ and verify that $\smash{\|\mathcal{P}_{N+3}(H_{\mathrm{B}}-zI)\mathcal{P}_N\hat\phi_n\|/\|\hat\phi_n\|}$ is small using interval arithmetic (S.M.). Panels (d) and (e) of \cref{fig:pseudospectra} show the resulting eigenfunction oscillations; panel (f) confirms Eq. (\ref{HB_asympt}), giving $C_{\mathrm{B}} \approx 0.176$.

\begin{figure}
\centering
\includegraphics[width=0.43\textwidth]{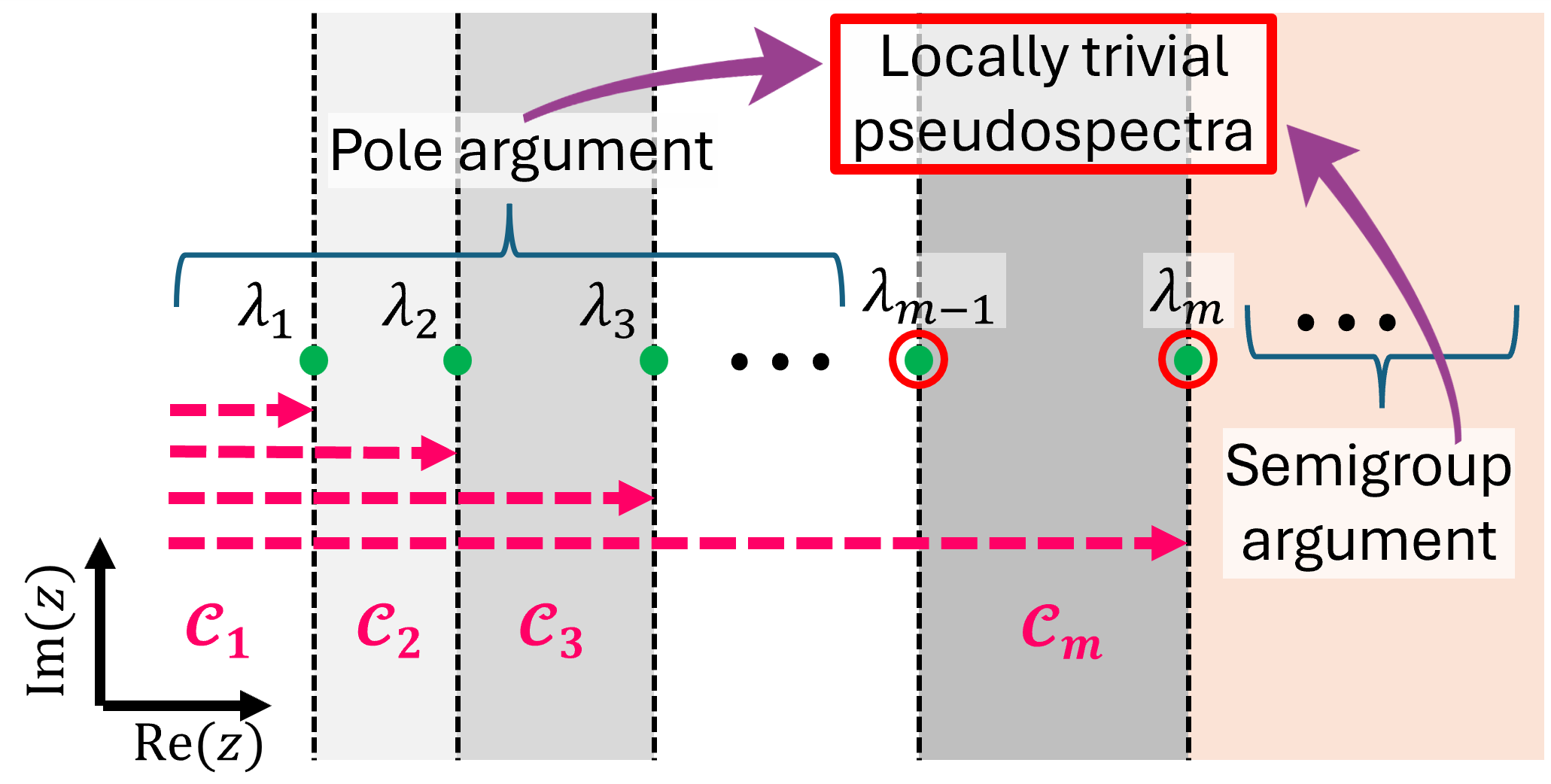}
\caption{The proof of Eq.~\eqref{eqmain_theorem} (S.M.) establishes \textit{LTP} on vertical strips (shaded gray) between eigenvalues. To the left of each strip, we bound eigenvalue contributions using pole bounds $\kappa_m/|\lambda_m-z|$ and truncated operators $\mathcal{C}_m$. To the right, we bound the semigroup generated by $-H_{\mathrm{B}}$. This approach extends to other operators with \textit{LTP}.}
\label{fig:thm_proof}
\end{figure}

\begin{figure}
\centering
\includegraphics[width=0.42\textwidth]{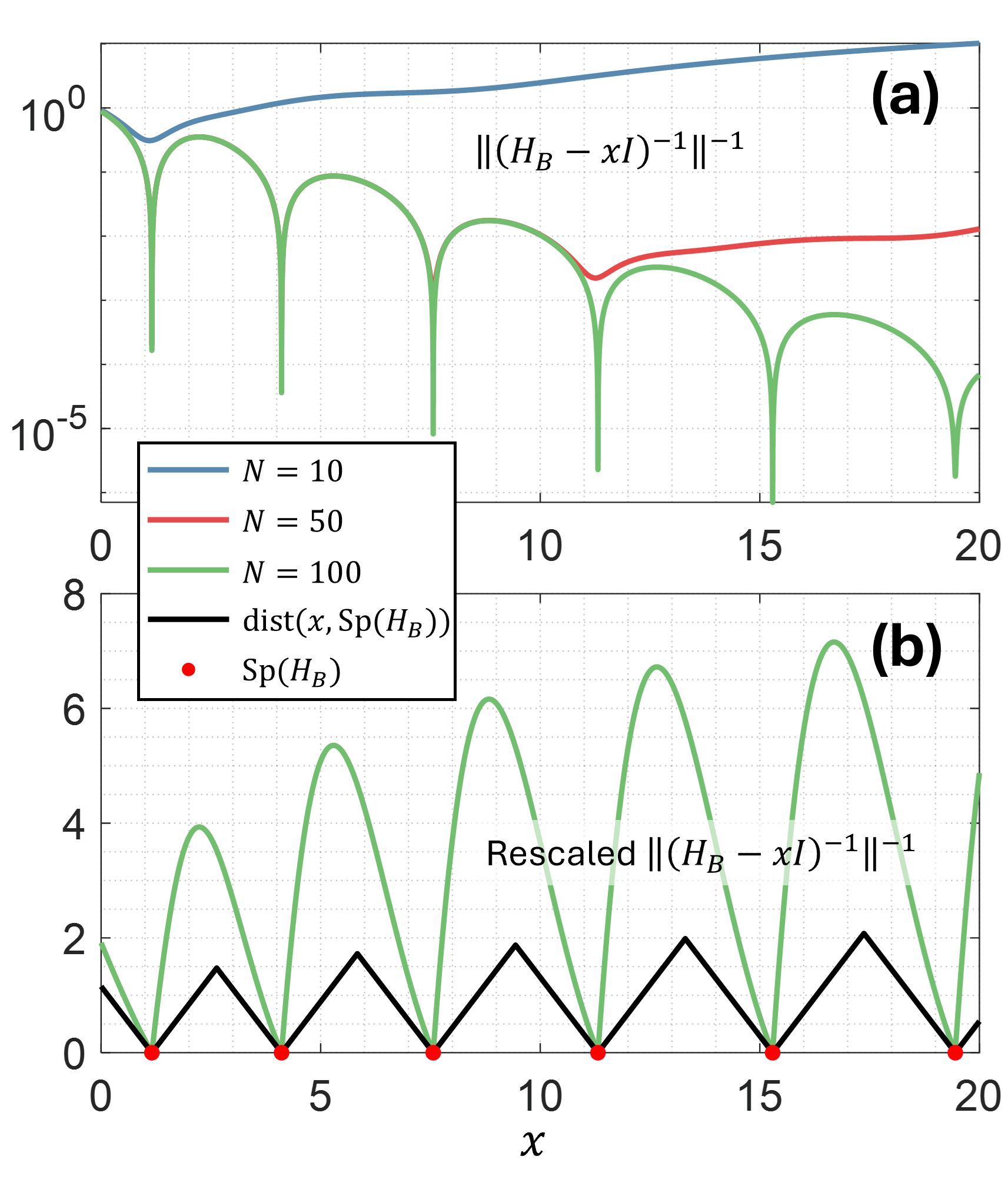}
\caption{\textbf{(a)}: Approximations of $\|(H_{\mathrm{B}}-xI)^{-1}\|^{-1}$ using $\mathrm{Sp}_\epsilon(\mathcal{P}_{N+3}H_{\mathrm{B}}\mathcal{P}_N)$, which converge from above as $N\rightarrow\infty$, free of spurious eigenvalues. \textbf{(b)}: Rescaled $\|(H_{\mathrm{B}}-xI)^{-1}\|^{-1}$, which are upper bounds on $\mathrm{dist}(x,\mathrm{Sp}(H_{\mathrm{B}}))$ used to compute $\mathrm{Sp}(H_{\mathrm{B}})$ with error bounds.}
\label{fig:rescale}
\end{figure}

\vspace{1mm}
\textit{Calculating the spectrum of $H_{\mathrm{B}}$.---}To compute the constants $C_K$ in Eq. (\ref{eq:locally_trivial_pseudopsectrum}) for $H = H_{\mathrm{B}}$ and address \textbf{(R2)} via \textit{LTP}, we truncate the eigenfunction expansion of $\mathcal{C}$ to $\mathcal{C}_m(x,y)=\sum_{n=1}^m \phi_n(x)\phi_n(y)$. From \cref{fig:pseudospectra}, $\kappa_n\leq \exp(\pi n/\sqrt{3})$, allowing a local bound on the difference between $\langle \cdot,\cdot \rangle$ and $\langle \cdot,\cdot\rangle_{\mathcal{CPT}}$. For $z\in\mathbb{C}\backslash\mathrm{Sp}(H_{\mathrm{B}})$, we bound contributions from eigenvalues $\lambda_n<\mathrm{Re}(z)$ using $\kappa_n/|\lambda_n-z|$, and for $\lambda_n>\mathrm{Re}(z)$, we bound the semigroup generated by $-H_{\mathrm{B}}$ on the corresponding eigenspaces. This requires asymptotic analysis of $\lambda_n$ to control the tail in $\mathcal{C}\approx\mathcal{C}_m$. The Hille--Yosida theorem \cite{pazy2012semigroups} then yields a bound on $\|(H_{\mathrm{B}}-zI)^{-1}\|$.

\cref{fig:thm_proof} illustrates this argument.
The end result (see S.M. for the full proof) is explicit bounds for \textit{LTP} of $H_{\mathrm{B}}$. Namely, if $\lambda_{m-1}< \mathrm{Re}(z)< \lambda_m$, then
\begin{equation}\label{eqmain_theorem}
\setlength\abovedisplayskip{4pt}\setlength\belowdisplayskip{4pt}
\|(H_{\mathrm{B}}{-}zI)^{-1}\|{\leq} \frac{\exp((m{-}1)\frac{\pi}{\sqrt{3}})}{|\lambda_{m-1}-z|}{+}\frac{\exp(m\frac{\pi}{\sqrt{3}})}{|\lambda_{m}-z|}{+}c_m,
\end{equation}
$c_m{=}{\exp\{(m{+}1)\frac{\pi}{\sqrt{3}}{+}[2 \Gamma(\frac{11}{6})\sqrt{\frac{\pi}{3}}/\Gamma(\frac{4}{3})(m{+}1)]^{\frac{6}{5}}\}}/{14}$.

Combining this with rectangular truncations (\cref{fig:rect_truncation}), we can now compute the eigenvalues of $H_{\mathrm{B}}$ with explicit error bounds. \cref{fig:rescale} (top) shows approximations of $\|(H_{\mathrm{B}}-zI)^{-1}\|^{-1}$ using rectangular truncations and $\mathrm{Sp}_\epsilon(\mathcal{P}_{N+3}H_{\mathrm{B}}\mathcal{P}_N)$. As $N\rightarrow\infty$, these converge from above and avoid spurious eigenvalues. From Eq. \eqref{eqmain_theorem},\begin{equation}
\label{inversion_formula}
\mathrm{dist}(z,\mathrm{Sp}(H_{\mathrm{B}}))\leq \frac{2\exp(m\frac{\pi}{\sqrt{3}})\|(H_{\mathrm{B}}-zI)^{-1}\|^{-1}}{1-c_m\|(H_{\mathrm{B}}-zI)^{-1}\|^{-1}}.
\end{equation}
\cref{fig:rescale} (bottom) shows the effect of this rescaling: the rescaled functions converge to an upper bound on $\mathrm{dist}(z,\mathrm{Sp}(H_{\mathrm{B}}))$, with local minima at the eigenvalues.

We verify eigenvalues in two steps. First, using interval bisection, we locate local minima $z_n$ of the rescaled $\|(H_{\mathrm{B}}-zI)^{-1}\|^{-1}$, obtaining candidate eigenfunctions $\smash{\hat{\phi}_n}$ from the corresponding right-singular vectors. Eq. (\ref{eqmain_theorem}) also predicts the required guard digits. Second, we use interval arithmetic to bound $\smash{\|(H_{\mathrm{B}}-z_nI)\hat{\phi}_n}\|$, giving a rigorous upper bound on $\|(H_{\mathrm{B}}-zI)^{-1}\|^{-1}$. Applying Eq. (\ref{eqmain_theorem}) again yields a bound on $|z_n - \lambda_n|$. Verified eigenvalues to 30 digits are shown in \cref{tab:cubic}. In this manner, we can also compute eigenvectors (S.M.).

\begin{table}[t!]
\centering
\caption{Eigenvalues of $H_{\mathrm{B}}$ with absolute error $<10^{-31}$.}
\begin{ruledtabular}
\begin{tabular}{rl}
\multicolumn{1}{r}{$n$} & \multicolumn{1}{c}{Verified eigenvalue $\lambda_n$}\\
\colrule
$1$ &    \hphantom{00}$1.156$ $267$ $071$ $988$ $113$ $293$ $799$ $219$ $177$ $999$ $9$ \\
$2$ &    \hphantom{00}$4.109$ $228$ $752$ $809$ $651$ $535$ $843$ $668$ $478$ $561$ $3$ \\
$3$ &    \hphantom{00}$7.562$ $273$ $854$ $978$ $828$ $041$ $351$ $809$ $110$ $631$ $4$ \\
$4$ &    \hphantom{0}$11.314$ $421$ $820$ $195$ $804$ $402$ $233$ $783$ $948$ $426$ $9$ \\
$5$ &    \hphantom{0}$15.291$ $553$ $750$ $392$ $532$ $388$ $181$ $630$ $791$ $751$ $9$ \\
$6$ &    \hphantom{0}$19.451$ $529$ $130$ $691$ $728$ $314$ $686$ $111$ $714$ $104$ $4$ \\
$7$ &    \hphantom{0}$23.766$ $740$ $435$ $485$ $819$ $131$ $558$ $025$ $968$ $789$ $9$ \\
$8$ &  	 \hphantom{0}$28.217$ $524$ $972$ $981$ $193$ $297$ $595$ $053$ $878$ $268$ $9$ \\
$9$ &    \hphantom{0}$32.789$ $082$ $781$ $862$ $957$ $492$ $447$ $371$ $485$ $046$ $3$ \\ 
$10$ &   \hphantom{0}$37.469$ $825$ $360$ $516$ $046$ $866$ $428$ $873$ $594$ $530$ $5$ \\ 
$100$ &             $627.694$ $712$ $248$ $436$ $511$ $352$ $673$ $702$ $901$ $153$ $6$ \\
\end{tabular}
\end{ruledtabular}
\label{tab:cubic}\vspace{-4mm}
\end{table}

\vspace{1mm}
\textit{Outlook and discussion.---}We proved that two key physical assumptions are necessary for computing spectra and introduced the concept of \textit{LTP}. The first assumption, concerning interaction terms, can be enforced via rectangular truncations of Hamiltonians and applies to systems with locally bounded interaction range, yielding sparse matrices. The second addresses non-Hermiticity and is realised through \textit{LTP}, which is broadly applicable and, in $\mathcal{PT}$-symmetric systems, serves as a local approximation of the hidden symmetry operator $\mathcal{C}$.

As an illustrative example, we computed the first verified eigenvalues of the imaginary cubic oscillator. To address \textbf{(R1)}, we used rectangular truncations, approximating the spectrum via singular values to avoid spurious modes. To address \textbf{(R2)}, we demonstrated \textit{LTP} through asymptotic and semigroup analysis.

The method extends to a wide range of operators, as illustrated by the physically motivated examples in \cref{fig:other_examples} (End Matter), which exhibit diverse spectral features and highlight the broad applicability of our approach and \textit{LTP}. For instance, beyond studying $H_{\mathrm{B}}$ on $L^2(\mathbb{R})$, we can consider operators on complex domains, such as those with potentials $x^2(ix)^\nu$ \cite{bender1999PT}, including the upside-down quartic (1st row). Additional examples include operators unbounded below (2nd row), with non-real spectra (3rd row), broken $\mathcal{PT}$-symmetry (4th row), and complex spectra with accumulation points (5th row). In each case, the Hermite function representation is replaced by another representation to overcome \textbf{(R1)}.

We can also deal with $H$ that have long-range interactions in two ways. First, by controlling the tail of rectangular truncations: e.g., if hopping terms decay exponentially, one can choose a truncation $\alpha(N)$ such that $\|(I-\mathcal{P}_{\alpha(N)})H\mathcal{P}_N\|\leq c_N$ for exponentially decaying $c_N$, allowing this error to be absorbed into the approximation. An example is given in the S.M.. Second, we may exploit additional information from the matrix coefficients of $H^*H$, as in the third row of \cref{fig:other_examples}. We outline the procedure in the End Matter.

We can also generalize Eq. (\ref{eq:locally_trivial_pseudopsectrum}) to handle non-simple eigenvalues by exploiting local algebraic relationships between $\|(H-zI)^{-1}\|^{-1}$ and $\mathrm{dist}(z,\mathrm{Sp}(H))$, with inversion formulas like Eq. (\ref{inversion_formula}). For instance, if an eigenvalue has algebraic multiplicity $p$, then a higher-order generalization of LTP near that eigenvalue yields
\begin{equation}\label{generalizedLTP}
\setlength\abovedisplayskip{2pt}\setlength\belowdisplayskip{2pt}
\mathrm{Sp}_{\epsilon}(H)\cap K\subset\{z\in\mathbb{C}:\mathrm{dist}(z,\mathrm{Sp}(H))\leq C_K\epsilon^{1/p}\}
\end{equation}
This is particularly useful for analyzing exceptional points, as illustrated in the End Matter.

Additionally, our semigroup approach applies when $\alpha H + \beta I$ generates a strongly continuous semigroup for appropriate scalars $\alpha, \beta$. Future work will also explore non-Hermitian operators with continuous spectra \cite{wen2020symmetric}.

Accurate eigenvalue computations are essential for the stability and efficiency of non-Hermitian systems, as small errors can cause significant deviations in physical behavior and compromise safety. Beyond the applications in the introduction, our method supports optimization in areas such as control for non-invasive imaging and wireless medical devices \cite{ye2021multi}. Moreover, gain–loss system phenomena extend beyond $\mathcal{PT}$-symmetry. Our approach to accurately compute eigenvalues and eigenvectors enables analysis of complex systems such as power grids, where balancing supply and demand poses challenges like intermittent generation and the dual roles of energy and information transmission \cite{powergrid, infopowergrid}.

We established a foundational link between quantum mechanics and computation by introducing \textit{LTP} and proving general limitations on spectral calculations. By understanding these limitations, we overcame them to compute verified eigenvalues of the imaginary cubic oscillator and other non-Hermitian operators. This work has both foundational and applied implications for non-Hermitian physics.

\footnotesize
\begin{acknowledgments}
The authors would like to thank the LMS for the Research in Pairs grant that facilitated the research in this project. M. C. was supported by Isaac Newton Trust Grant No. LEAG/929. We thank the referees for their helpful suggestions and comments.
\end{acknowledgments}

\nocite{colbrook3,kato1949upperH,zimmermann1995variational,barrenechea2014finite,press2007numerical,dondl2017bound,tai2006simpleness,davies2005semigroup}

\bibliographystyle{apsrev4-2}
\bibliography{cubic_oscillator_final_v3}

@Article{dondl2017bound,
  author    = {Dondl, Patrick W. and Dorey, Patrick and R{\"o}sler, Frank},
  journal   = {Appl. Math. Res. eXpress},
  title     = {[sixth reference in Supplemental Material not already in paper] A bound on the pseudospectrum for a class of non-normal {S}chr{\"o}dinger operators},
  year      = {2017},
  issn      = {1687-1197},
  month     = dec,
  number    = {2},
  pages     = {271--296},
  volume    = {2017},
  doi       = {10.1093/amrx/abw011},
  publisher = {Oxford University Press (OUP)},
}

@Article{davies2005semigroup,
  author    = {Davies, E. Brian},
  journal   = {J. Operat. Th.},
  title     = {[eighth reference in Supplemental Material not already in paper] Semigroup growth bounds},
  year      = {2005},
  number    = {2},
  pages     = {225--249},
  volume    = {53},
  publisher = {JSTOR},
}

@Article{PhysRevLett.122.250201,
  author    = {Colbrook, Matthew J. and Roman, Bogdan and Hansen, Anders C.},
  journal   = {Phys. Rev. Lett.},
  title     = {How to compute spectra with error control},
  year      = {2019},
  issn      = {1079-7114},
  month     = jun,
  number    = {25},
  pages     = {250201},
  volume    = {122},
  doi       = {10.1103/physrevlett.122.250201},
  publisher = {American Physical Society (APS)},
}

@Article{BenderNonHermitianHamiltoniansSense,
  author    = {Bender, Carl M.},
  journal   = {Rep. Prog. Phys.},
  title     = {Making sense of non-{H}ermitian {H}amiltonians},
  year      = {2007},
  issn      = {1361-6633},
  number    = {6},
  pages     = {947--1018},
  volume    = {70},
  doi       = {10.1088/0034-4885/70/6/r03},
  publisher = {IOP Publishing},
}

@Article{PhysRevD.86.121702,
  author    = {Siegl, Petr and Krej{\v{c}}i{\v{r}}{\'\i}k, David},
  journal   = {Phys. Rev. D},
  title     = {On the metric operator for the imaginary cubic oscillator},
  year      = {2012},
  issn      = {1550-2368},
  number    = {12},
  pages     = {121702},
  volume    = {86},
  doi       = {10.1103/physrevd.86.121702},
  publisher = {American Physical Society (APS)},
}

@Article{Mostafazadeh_2006,
  author    = {Mostafazadeh, Ali},
  journal   = {J. Phys. A Math. Gen.},
  title     = {Metric operator in pseudo-{H}ermitian quantum mechanics and the imaginary cubic potential},
  year      = {2006},
  issn      = {1361-6447},
  number    = {32},
  pages     = {10171--10188},
  volume    = {39},
  doi       = {10.1088/0305-4470/39/32/s18},
  publisher = {IOP Publishing},
}

@Article{BenderComplexExtension,
  author    = {Bender, Carl M. and Brody, Dorje C. and Jones, Hugh F.},
  journal   = {Phys. Rev. Lett.},
  title     = {Complex extension of quantum mechanics},
  year      = {2002},
  issn      = {1079-7114},
  number    = {27},
  pages     = {270401},
  volume    = {89},
  doi       = {10.1103/physrevlett.89.270401},
  publisher = {American Physical Society (APS)},
}

@Article{BenderRealSpectrainNon-HermitianHamiltonians,
  author    = {Bender, Carl M. and Boettcher, Stefan},
  journal   = {Phys. Rev. Lett.},
  title     = {Real spectra in non-{H}ermitian {H}amiltonians having {PT} Symmetry},
  year      = {1998},
  issn      = {1079-7114},
  number    = {24},
  pages     = {5243--5246},
  volume    = {80},
  doi       = {10.1103/physrevlett.80.5243},
  publisher = {American Physical Society (APS)},
}

@Article{BenderHamiltonianHermitian,
  author    = {Bender, Carl M. and Brody, Dorje C. and Jones, Hugh F.},
  journal   = {Am. J. Phys.},
  title     = {Must a {H}amiltonian be {H}ermitian?},
  year      = {2003},
  issn      = {1943-2909},
  number    = {11},
  pages     = {1095--1102},
  volume    = {71},
  doi       = {10.1119/1.1574043},
  publisher = {American Association of Physics Teachers (AAPT)},
}

@Article{CarlMBenderHiddenSymmetry,
  author    = {Bender, Carl M. and Meisinger, Peter N. and Wang, Qinghai},
  journal   = {J. Phys. A Math. Gen.},
  title     = {Calculation of the hidden symmetry operator in {PT}-symmetric quantum mechanics},
  year      = {2003},
  issn      = {0305-4470},
  number    = {7},
  pages     = {1973--1983},
  volume    = {36},
  doi       = {10.1088/0305-4470/36/7/312},
  publisher = {IOP Publishing},
}

@Article{BenderCubicInteraction,
  author    = {Bender, Carl M. and Brody, Dorje C. and Jones, Hugh F.},
  journal   = {Phys. Rev. D},
  title     = {Extension of {PT}-symmetric quantum mechanics to quantum field theory with cubic interaction},
  year      = {2004},
  issn      = {1550-2368},
  month     = jul,
  number    = {2},
  pages     = {025001},
  volume    = {70},
  doi       = {10.1103/physrevd.70.025001},
  publisher = {American Physical Society (APS)},
}

@Article{KSTV,
  author    = {Krej\v{c}i\v{r}\'{\i}k, D. and others},
  journal   = {J. Math. Phys.},
  title     = {Pseudospectra in non-{H}ermitian quantum mechanics},
  year      = {2015},
  issn      = {1089-7658},
  number    = {10},
  volume    = {56},
  doi       = {10.1063/1.4934378},
  publisher = {AIP Publishing},
}

@Article{BoegliSieglTretter,
  author    = {B{\"o}gli, Sabine and Siegl, Petr and Tretter, Christiane},
  journal   = {Commun. Partial Differ. Equ.},
  title     = {Approximations of spectra of {S}chr{\"o}dinger operators with complex potentials on {$\mathbb{R}^d$}},
  year      = {2017},
  issn      = {1532-4133},
  month     = jun,
  number    = {7},
  pages     = {1001--1041},
  volume    = {42},
  doi       = {10.1080/03605302.2017.1330342},
  publisher = {Informa UK Limited},
}

@Book{pazy2012semigroups,
  author    = {Pazy, A.},
  publisher = {Springer New York},
  title     = {Semigroups of {L}inear {O}perators and {A}pplications to {P}artial {D}ifferential {E}quations},
  year      = {1983},
  isbn      = {9781461255611},
  volume    = {44},
  doi       = {10.1007/978-1-4612-5561-1},
  issn      = {0066-5452},
  journal   = {Applied Mathematical Sciences},
}

@Book{trefethen2005spectra,
  author    = {Trefethen, Lloyd N. and Embree, Mark},
  publisher = {Princeton University Press},
  title     = {Spectra and pseudospectra},
  year      = {2005},
  isbn      = {9780691213101},
  month     = jan,
  doi       = {10.1515/9780691213101},
}

@article{bender1999variational,
  title={Variational ansatz for {PT}-symmetric quantum mechanics},
  author={Bender, Carl M. and others},
  journal={Phys. Lett. A},
  volume={259},
  number={3-4},
  pages={224--231},
  year={1999},
  publisher={Elsevier}
}

@article{handy2003moment,
  title={Moment problem quantization within a generalized scalet-{W}igner (auto-scaling) transform representation},
  author={Handy, CR and others},
  journal={J. Phys. A Math. Gen.},
  volume={36},
  number={6},
  pages={1623},
  year={2003},
  publisher={IOP Publishing}
}

@book{bender2019pt,
  title={PT symmetry: In quantum and classical physics},
  author={Bender, Carl M},
  year={2019},
  publisher={World Scientific}
}

@article{PhysRevLett.103.093902,
  title = {Observation of $\mathcal{P}\mathcal{T}$-Symmetry Breaking in Complex Optical Potentials},
  author = {Guo, A. and others},
  journal = {Phys. Rev. Lett.},
  volume = {103},
  issue = {9},
  pages = {093902},
  numpages = {4},
  year = {2009},
  month = {Aug},
  publisher = {American Physical Society},
  doi = {10.1103/PhysRevLett.103.093902},
  url = {https://link.aps.org/doi/10.1103/PhysRevLett.103.093902}
}

@article{ruter2010observation,
  title={Observation of parity--time symmetry in optics},
  author={R{\"u}ter, Christian E and others},
  journal={Nat. Phys.},
  volume={6},
  number={3},
  pages={192--195},
  year={2010},
  publisher={Nature Publishing Group UK London}
}

@article{feng2011nonreciprocal,
  title={Nonreciprocal light propagation in a silicon photonic circuit},
  author={Feng, Liang and others},
  journal={Science},
  volume={333},
  number={6043},
  pages={729--733},
  year={2011},
  publisher={American Association for the Advancement of Science}
}

@Article{regensburger2012parity,
  author    = {Regensburger, Alois and others},
  journal   = {Nature},
  title     = {Parity-time synthetic photonic lattices},
  year      = {2012},
  number    = {7410},
  pages     = {167--171},
  volume    = {488},
  doi       = {10.1038/nature11298},
  publisher = {Springer Science and Business Media {LLC}},
}

@article{schindler2011experimental,
  title={Experimental study of active LRC circuits with PT symmetries},
  author={Schindler, Joseph and others},
  journal={Phys. Rev. A},
  volume={84},
  number={4},
  pages={040101},
  year={2011},
  publisher={APS}
}

@article{bittner2012pt,
  title={{PT} symmetry and spontaneous symmetry breaking in a microwave billiard},
  author={Bittner, S and others},
  journal={Phys. Rev. Lett.},
  volume={108},
  number={2},
  pages={024101},
  year={2012},
  publisher={APS}
}

@article{peng2014parity,
  title={Parity--time-symmetric whispering-gallery microcavities},
  author={Peng, Bo and others},
  journal={Nat. Phys.},
  volume={10},
  number={5},
  pages={394--398},
  year={2014},
  publisher={Nature Publishing Group UK London}
}

@article{assawaworrarit2017robust,
  title={Robust wireless power transfer using a nonlinear parity--time-symmetric circuit},
  author={Assawaworrarit, Sid and Yu, Xiaofang and Fan, Shanhui},
  journal={Nature},
  volume={546},
  number={7658},
  pages={387--390},
  year={2017},
  publisher={Nature Publishing Group UK London}
}

@article{feng2014single,
  title={Single-mode laser by parity-time symmetry breaking},
  author={Feng, Liang and others},
  journal={Science},
  volume={346},
  number={6212},
  pages={972--975},
  year={2014},
  publisher={American Association for the Advancement of Science}
}

@article{hodaei2014parity,
  title={Parity-time--symmetric microring lasers},
  author={Hodaei, Hossein and others},
  journal={Science},
  volume={346},
  number={6212},
  pages={975--978},
  year={2014},
  publisher={American Association for the Advancement of Science}
}

@article{chen2016pt,
  title={PT symmetry and singularity-enhanced sensing based on photoexcited graphene metasurfaces},
  author={Chen, Pai-Yen and Jung, Jeil},
  journal={Phys. Rev. Applied},
  volume={5},
  number={6},
  pages={064018},
  year={2016},
  publisher={APS}
}

@article{liu2016metrology,
  title={Metrology with PT-symmetric cavities: enhanced sensitivity near the PT-phase transition},
  author={Liu, Zhong-Peng and others},
  journal={Phys. Rev. Lett.},
  volume={117},
  number={11},
  pages={110802},
  year={2016},
  publisher={APS}
}

@article{soley2023experimentally,
  title={Experimentally realizable pt phase transitions in reflectionless quantum scattering},
  author={Soley, Micheline B and Bender, Carl M and Stone, A Douglas},
  journal={Phys. Rev. Lett.},
  volume={130},
  number={25},
  pages={250404},
  year={2023},
  publisher={APS}
}

@Article{dorey2001spectral,
  author    = {Dorey, Patrick and Dunning, Clare and Tateo, Roberto},
  journal   = {J. Phys. A Math. Gen.},
  title     = {Spectral equivalences, {B}ethe ansatz equations, and reality properties in {PT}-symmetric quantum mechanics},
  year      = {2001},
  number    = {28},
  pages     = {5679--5704},
  volume    = {34},
  doi       = {10.1088/0305-4470/34/28/305},
  publisher = {{IOP} Publishing},
}

@Article{caliceti1980perturbation,
  author    = {Caliceti, E. and Graffi, S. and Maioli, M.},
  journal   = {Commun. Math. Phys.},
  title     = {Perturbation theory of odd anharmonic oscillators},
  year      = {1980},
  number    = {1},
  pages     = {51--66},
  volume    = {75},
  doi       = {10.1007/bf01962591},
  publisher = {Springer Science and Business Media {LLC}},
}

@Article{colbrook3,
  author    = {Colbrook, Matthew J. and Hansen, Anders C.},
  journal   = {J. Eur. Math. Soc.},
  title     = {[first reference in Supplemental Material not already in paper] The foundations of spectral computations via the solvability complexity index hierarchy},
  year      = {2022},
  issn      = {1435-9855},
  number    = {12},
  pages     = {4639--4728},
  volume    = {25},
  doi       = {10.4171/jems/1289},
  publisher = {European Mathematical Society - EMS - Publishing House GmbH},
}

@article{chandler2024spectral,
  title={On spectral inclusion sets and computing the spectra and pseudospectra of bounded linear operators},
  author={Chandler-Wilde, Simon N and Chonchaiya, Ratchanikorn and Lindner, Marko},
  journal={J. Spectr. Theory},
  volume={14},
  number={2},
  pages={719--804},
  year={2024},
  publisher={EMS Press}
}

@Article{henry2014spectralcubic,
  author       = {Henry, Rapha{\"e}l},
  journal      = {Annales Henri Poincar{\'{e}}},
  title        = {Spectral projections of the complex cubic oscillator},
  year         = {2014},
  number       = {10},
  pages        = {2025--2043},
  volume       = {15},
  doi          = {10.1007/s00023-013-0292-2},
  organization = {Springer},
  publisher    = {Springer Science and Business Media {LLC}},
}

@Article{barrenechea2014finite,
  author    = {Barrenechea, Gabriel R. and Boulton, Lyonell and Boussaid, Nabile},
  journal   = {SIAM J. Sci. Comput.},
  title     = {[fourth reference in Supplemental Material not already in paper] Finite element eigenvalue enclosures for the {M}axwell operator},
  year      = {2014},
  number    = {6},
  pages     = {A2887--A2906},
  volume    = {36},
  doi       = {10.1137/140957810},
  publisher = {Society for Industrial {\&} Applied Mathematics ({SIAM})},
}

@Article{kato1949upperH,
  author    = {Kato, Tosio},
  journal   = {J. Phys. Soc. Jpn.},
  title     = {[second reference in Supplemental Material not already in paper] On the upper and lower bounds of eigenvalues},
  year      = {1949},
  number    = {4-6},
  pages     = {334--339},
  volume    = {4},
  doi       = {10.1143/jpsj.4.334},
  publisher = {Physical Society of Japan},
}

@Article{zimmermann1995variational,
  author    = {S. Zimmermann and U. Mertins},
  journal   = {Z. Anal. Anwend.},
  title     = {[third reference in Supplemental Material not already in paper] Variational bounds to eigenvalues of self-adjoint eigenvalue problems with arbitrary spectrum},
  year      = {1995},
  number    = {2},
  pages     = {327--345},
  volume    = {14},
  doi       = {10.4171/zaa/677},
  publisher = {European Mathematical Society - {EMS} - Publishing House {GmbH}},
}

@Article{Hansen_JAMS,
  author    = {Hansen, Anders C.},
  journal   = {J. Am. Math. Soc.},
  title     = {On the solvability complexity index, the {$n$}-pseudospectrum and approximations of spectra of operators},
  year      = {2011},
  issn      = {1088-6834},
  number    = {1},
  pages     = {81--124},
  volume    = {24},
  doi       = {10.1090/s0894-0347-2010-00676-5},
  publisher = {American Mathematical Society (AMS)},
}

@Article{ben2015can,
  author        = {Ben-Artzi, Jonathan and others},
  title         = {Computing spectra - {O}n the solvability complexity index hierarchy and towers of algorithms},
  year          = {2020},
  doi           = {10.48550/ARXIV.1508.03280},
  eprint        = {1508.03280},
  journal     = {arXiv},
}

@PhdThesis{colbrook2020PhD,
  author = {Colbrook, Matthew J.},
  school = {University of Cambridge},
  title  = {The foundations of infinite-dimensional spectral computations},
  year   = {2020},
}

@Article{tai2006simpleness,
  author    = {Trinh Duc Tai},
  journal   = {J. Differential Equations},
  title     = {[seventh reference in Supplemental Material not already in paper] On the simpleness of zeros of {S}tokes multipliers},
  year      = {2006},
  number    = {2},
  pages     = {351--366},
  volume    = {223},
  doi       = {10.1016/j.jde.2005.07.020},
  publisher = {Elsevier},
}

@Book{kato2013perturbation,
  author    = {Tosio Kato},
  publisher = {Springer Berlin Heidelberg},
  title     = {Perturbation {T}heory for {L}inear {O}perators},
  year      = {1995},
  edition   = {Second},
  volume    = {132},
  doi       = {10.1007/978-3-642-66282-9},
}

@book{press2007numerical,
  title={[fifth reference in Supplemental Material not already in paper] Numerical recipes 3rd edition: {T}he art of scientific computing},
  author={Press, William H. and Vetterling, William T. and Teukolsky, Saul A. and Flannery, Brian P.},
  year={2007},
  publisher={Cambridge university press}
}

@article{wen2020symmetric,
  title={{PT}-symmetric potentials having continuous spectra},
  author={Wen, Zichao and Bender, Carl M.},
  journal={J. Phys. A-Math.},
  volume={53},
  number={37},
  pages={375302},
  year={2020},
  publisher={IOP Publishing}
}

@article{bender1999PT,
  title={{PT}-symmetric quantum mechanics},
  author={Bender, Carl M. and Boettcher, Stefan and Meisinger, Peter N.},
  journal={J. Math. Phys.},
  volume={40},
  number={5},
  pages={2201--2229},
  year={1999},
  publisher={American Institute of Physics}
}

@article{xiao2021observation,
  title={Observation of non-{B}loch parity-time symmetry and exceptional points},
  author={Xiao, Lei and others},
  journal={Phys. Rev. Lett.},
  volume={126},
  number={23},
  pages={230402},
  year={2021},
  publisher={APS}
}

@article{yang2022observation,
  title={Observation of transient parity-time symmetry in electronic systems},
  author={Yang, Xin and others},
  journal={Phys. Rev. Lett.},
  volume={128},
  number={6},
  pages={065701},
  year={2022},
  publisher={APS}
}

@article{li2024experimental,
  title={Experimental Demonstration of Controllable {PT} and anti-{PT} Coupling in a non-{H}ermitian Metamaterial},
  author={Li, Chang and others},
  journal={Phys. Rev. Lett.},
  volume={132},
  number={15},
  pages={156601},
  year={2024},
  publisher={APS}
}

@article{zhang2020synthetic,
  title={Synthetic anti-{PT} symmetry in a single microcavity},
  author={Zhang, Fangxing and others},
  journal={Phys. Rev. Lett.},
  volume={124},
  number={5},
  pages={053901},
  year={2020},
  publisher={APS}
}

@article{el2018non,
  title={Non-Hermitian physics and {PT} symmetry},
  author={El-Ganainy, Ramy and others},
  journal={Nat. Phys.},
  volume={14},
  number={1},
  pages={11--19},
  year={2018},
  publisher={Nature Publishing Group UK London}
}

@article{xiao2017observation,
  title={Observation of topological edge states in parity--time-symmetric quantum walks},
  author={Xiao, L and others},
  journal={Nat. Phys.},
  volume={13},
  number={11},
  pages={1117--1123},
  year={2017},
  publisher={Nature Publishing Group UK London}
}

@article{weimann2017topologically,
  title={Topologically protected bound states in photonic parity--time-symmetric crystals},
  author={Weimann, Steffen and others},
  journal={Nat. Mater.},
  volume={16},
  number={4},
  pages={433--438},
  year={2017},
  publisher={Nature Publishing Group UK London}
}

@article{zhang2016observation,
  title={Observation of parity-time symmetry in optically induced atomic lattices},
  author={Zhang, Zhaoyang and others},
  journal={Phys. Rev. Lett.},
  volume={117},
  number={12},
  pages={123601},
  year={2016},
  publisher={APS}
}

@article{chong2011pt,
  title={PT-symmetry breaking and laser-absorber modes in optical scattering systems},
  author={Chong, YD and Ge, Li and Stone, A Douglas},
  journal={Phys. Rev. Lett.},
  volume={106},
  number={9},
  pages={093902},
  year={2011},
  publisher={APS}
}

@article{musslimani2008optical,
  title={Optical solitons in PT periodic potentials},
  author={Musslimani, Ziad H and others},
  journal={Phys. Rev. Lett.},
  volume={100},
  number={3},
  pages={030402},
  year={2008},
  publisher={APS}
}

@book{dirac1981principles,
  title={The principles of quantum mechanics},
  author={Dirac, Paul Adrien Maurice},
  number={27},
  year={1981},
  publisher={Oxford university press}
}

@Book{von2018mathematical,
  author    = {von Neumann, John},
  publisher = {Princeton university press},
  title     = {Mathematical foundations of quantum mechanics},
  year      = {2018},
  volume    = {53},
}

@article{shi2016accessing,
  title={Accessing the exceptional points of parity-time symmetric acoustics},
  author={Shi, Chengzhi and others},
  journal={Nat. Commun.},
  volume={7},
  number={1},
  pages={11110},
  year={2016},
  publisher={Nature Publishing Group UK London}
}

@article{auregan2017pt,
  title={PT-symmetric scattering in flow duct acoustics},
  author={Aur{\'e}gan, Yves and Pagneux, Vincent},
  journal={Phys. Rev. Lett.},
  volume={118},
  number={17},
  pages={174301},
  year={2017},
  publisher={APS}
}

@article{chtchelkatchev2012stimulation,
  title={Stimulation of the fluctuation superconductivity by PT symmetry},
  author={Chtchelkatchev, NM and others},
  journal={Phys. Rev. Lett.},
  volume={109},
  number={15},
  pages={150405},
  year={2012},
  publisher={APS}
}

@article{bender2013observation,
  title={Observation of asymmetric transport in structures with active nonlinearities},
  author={Bender, Nicholas and others},
  journal={Phys. Rev. Lett.},
  volume={110},
  number={23},
  pages={234101},
  year={2013},
  publisher={APS}
}

@article{cao2022fully,
  title={Fully integrated parity--time-symmetric electronics},
  author={Cao, Weidong and others},
  journal={Nat. Nanotechnol.},
  volume={17},
  number={3},
  pages={262--268},
  year={2022},
  publisher={Nature Publishing Group UK London}
}

@article{bender2013observation2,
  title={Observation of PT phase transition in a simple mechanical system},
  author={Bender, Carl M and others},
  journal={Am. J. Phys.},
  volume={81},
  number={3},
  pages={173--179},
  year={2013},
  publisher={American Association of Physics Teachers}
}

@article{ye2021multi,
  title={Multi-band parity-time-symmetric wireless power transfer systems for {ISM}-band bio-implantable applications},
  author={Ye, Zhilu and Yang, Minye and Chen, Pai-Yen},
  journal={IEEE J. Electromagn. RF Microw. Med. Biol.},
  volume={6},
  number={2},
  pages={196--203},
  year={2021},
  publisher={IEEE}
}

@article{infopowergrid,
  title={For the grid and through the grid: The role of power line communications in the smart grid},
  author={Galli, Stefano and Scaglione, Anna and Wang, Zhifang},
  journal={Proc. IEEE},
  volume={99},
  number={6},
  pages={998--1027},
  year={2011},
  publisher={IEEE}
}

@article{powergrid,
  title={Models for the modern power grid},
  author={Nardelli, Pedro HJ and others},
  journal={Eur. Phys. J. Spec. Top.},
  volume={223},
  pages={2423--2437},
  year={2014},
  publisher={Springer}
}

@Article{Bender1999,
  author    = {C. M. Bender and K. A. Milton},
  journal   = {J. Phys. A: Math. Gen.	},
  title     = {A nonunitary version of massless quantum electrodynamics possessing a critical point},
  year      = {1999},
  issn      = {1361-6447},
  month     = jan,
  number    = {7},
  pages     = {L87--L92},
  volume    = {32},
  doi       = {10.1088/0305-4470/32/7/001},
  publisher = {IOP Publishing},
}

@Article{Bender2000,
  author    = {Bender, Carl M. and Milton, Kimball A. and Savage, Van M.},
  journal   = {Phys. Rev. D},
  title     = {Solution of Schwinger-Dyson equations for PTsymmetric quantum field theory},
  year      = {2000},
  issn      = {1089-4918},
  month     = sep,
  number    = {8},
  pages     = {085001},
  volume    = {62},
  doi       = {10.1103/physrevd.62.085001},
  publisher = {American Physical Society (APS)},
}

@Article{Bender1998,
  author    = {Bender, Carl M. and Milton, Kimball A.},
  journal   = {Phys. Rev. D},
  title     = {Model of supersymmetric quantum field theory with broken parity symmetry},
  year      = {1998},
  issn      = {1089-4918},
  month     = mar,
  number    = {6},
  pages     = {3595--3608},
  volume    = {57},
  doi       = {10.1103/physrevd.57.3595},
  publisher = {American Physical Society (APS)},
}

@Article{Bender2005,
  author    = {Bender, Carl M. and others},
  journal   = {Phys. Lett. B},
  title     = {PT-symmetric quantum electrodynamics},
  year      = {2005},
  issn      = {0370-2693},
  month     = apr,
  number    = {1–2},
  pages     = {97--104},
  volume    = {613},
  doi       = {10.1016/j.physletb.2005.03.032},
  publisher = {Elsevier BV},
}

@Article{Bender2001,
  author    = {Bender, Carl M and Meisinger, Peter N and Yang, Haitang},
  journal   = {Phys. Lett. D},
  title     = {Calculation of the one-point {G}reen’s function for a {$-g\varphi^4$} quantum field theory},
  year      = {2001},
  number    = {4},
  pages     = {045001},
  volume    = {63},
  publisher = {APS},
}

@InProceedings{Papkovich1940,
  author    = {Papkovich, PF},
  booktitle = {Dokl. Acad. Sci. USSR},
  title     = {Uber eine {F}orm der {L}{\"o}sung des byharmonischen {P}roblems f{\"u}r das {R}echteck},
  year      = {1940},
  pages     = {334--338},
  volume    = {27},
}

@Article{Fadle1940,
  author    = {Fadle, Johann},
  journal   = {Ingenieur-archiv},
  title     = {Die selbstspannungs-eigenwertfunktionen der quadratischen scheibe},
  year      = {1940},
  number    = {2},
  pages     = {125--149},
  volume    = {11},
  publisher = {Springer},
}

@Book{Shankar2007,
  author    = {Shankar, PN},
  publisher = {World Scientific},
  title     = {Slow Viscous Flows},
  year      = {2007},
}

@Article{Benilov2003,
  author    = {Benilov, ES and O'Brien, SBG and Sazonov, IA},
  journal   = {J. Fluid Mech.},
  title     = {A new type of instability: explosive disturbances in a liquid film inside a rotating horizontal cylinder},
  year      = {2003},
  pages     = {201--224},
  volume    = {497},
  publisher = {Cambridge University Press},
}

@Article{Fraternale2018,
  author    = {Fraternale, Federico and others},
  journal   = {Phys. Rev. E},
  title     = {Internal waves in sheared flows: Lower bound of the vorticity growth and propagation discontinuities in the parameter space},
  year      = {2018},
  number    = {6},
  pages     = {063102},
  volume    = {97},
  publisher = {APS},
}

@Article{Makris2008,
  author    = {Makris, Konstantinos G and others},
  journal   = {Phys. Rev. Lett.	},
  title     = {Beam dynamics in {PT} symmetric optical lattices},
  year      = {2008},
  number    = {10},
  pages     = {103904},
  volume    = {100},
  publisher = {APS},
}

@Article{Okuma2020,
  author    = {Okuma, Nobuyuki and others},
  journal   = {Phys. Rev. Lett.	},
  title     = {Topological origin of non-{H}ermitian skin effects},
  year      = {2020},
  number    = {8},
  pages     = {086801},
  volume    = {124},
  publisher = {APS},
}

@Article{Siegman1986,
  author  = {Siegman, Anthony E and others},
  journal = {Mill Valley, CA},
  title   = {Lasers university science books},
  year    = {1986},
  number  = {208},
  pages   = {169},
  volume  = {37},
}

@Article{Fox1963,
  author    = {Fox, AG and Li, Tingye},
  journal   = {Proc. IEEE},
  title     = {Modes in a maser interferometer with curved and tilted mirrors},
  year      = {1963},
  number    = {1},
  pages     = {80--89},
  volume    = {51},
  publisher = {IEEE},
}

@Article{Landau1976,
  author    = {Landau, HJ},
  journal   = {JOSA},
  title     = {Loss in unstable resonators},
  year      = {1976},
  number    = {6},
  pages     = {525--529},
  volume    = {66},
  publisher = {Optica Publishing Group},
}

@Article{Wright2020,
  author    = {Wright, Logan G. and others},
  journal   = {Nat. Phys.},
  title     = {Mechanisms of spatiotemporal mode-locking},
  year      = {2020},
  issn      = {1745-2481},
  month     = feb,
  number    = {5},
  pages     = {565--570},
  volume    = {16},
  doi       = {10.1038/s41567-020-0784-1},
  publisher = {Springer Science and Business Media LLC},
}

@Article{Nixon2013,
  author    = {Nixon, Micha and others},
  journal   = {Nat. Photonics},
  title     = {Real-time wavefront shaping through scattering media by all-optical feedback},
  year      = {2013},
  issn      = {1749-4893},
  month     = oct,
  number    = {11},
  pages     = {919--924},
  volume    = {7},
  doi       = {10.1038/nphoton.2013.248},
  publisher = {Springer Science and Business Media LLC},
}

@article{roch1996c,
  title={{$C*$}-algebra techniques in numerical analysis},
  author={Roch, Steffen and Silbermann, Bernd},
  journal={J. Operat. Th.},
  pages={241--280},
  year={1996},
  publisher={JSTOR}
}

@article{davies2000pseudospectra,
  title={Pseudospectra of differential operators},
  author={Davies, E Brian},
  journal={J. Operat. Th.},
  pages={243--262},
  year={2000},
  publisher={JSTOR}
}

@article{aslanyan2000spectral,
  title={Spectral instability for some {S}chr{\"o}dinger operators},
  author={Aslanyan, Anna and Davies, E Brian},
  journal={Numer. Math.},
  volume={85},
  pages={525--552},
  year={2000},
  publisher={Springer}
}

@Article{Hu2024,
  author    = {Hu, Yu-Min and Wang, Hong-Yi and Wang, Zhong and Song, Fei},
  journal   = {Physical Review Letters},
  title     = {Geometric Origin of Non-Bloch PT Symmetry Breaking},
  year      = {2024},
  issn      = {1079-7114},
  month     = jan,
  number    = {5},
  pages     = {050402},
  volume    = {132},
  doi       = {10.1103/physrevlett.132.050402},
  publisher = {American Physical Society (APS)},
}

@Article{Yang2024,
  author    = {Yang, Kang and Li, Zhi and König, J Lukas K and Rødland, Lukas and Stålhammar, Marcus and Bergholtz, Emil J},
  journal   = {Reports on Progress in Physics},
  title     = {Homotopy, symmetry, and non-Hermitian band topology},
  year      = {2024},
  issn      = {1361-6633},
  month     = jul,
  number    = {7},
  pages     = {078002},
  volume    = {87},
  doi       = {10.1088/1361-6633/ad4e64},
  publisher = {IOP Publishing},
}

@Article{Banerjee2023,
  author    = {Banerjee, Ayan and Sarkar, Ronika and Dey, Soumi and Narayan, Awadhesh},
  journal   = {Journal of Physics: Condensed Matter},
  title     = {Non-Hermitian topological phases: principles and prospects},
  year      = {2023},
  issn      = {1361-648X},
  month     = may,
  number    = {33},
  pages     = {333001},
  volume    = {35},
  doi       = {10.1088/1361-648x/acd1cb},
  publisher = {IOP Publishing},
}

@Article{Ohnmacht2025,
  author    = {Ohnmacht, David Christian and Wilhelm, Valentin and Weisbrich, Hannes and Belzig, Wolfgang},
  journal   = {Physical Review Letters},
  title     = {Non-Hermitian Topology in Multiterminal Superconducting Junctions},
  year      = {2025},
  issn      = {1079-7114},
  month     = apr,
  number    = {15},
  pages     = {156601},
  volume    = {134},
  doi       = {10.1103/physrevlett.134.156601},
  publisher = {American Physical Society (APS)},
}

@Article{Bergholtz2021,
  author    = {Bergholtz, Emil J. and Budich, Jan Carl and Kunst, Flore K.},
  journal   = {Reviews of Modern Physics},
  title     = {Exceptional topology of non-Hermitian systems},
  year      = {2021},
  issn      = {1539-0756},
  month     = feb,
  number    = {1},
  pages     = {015005},
  volume    = {93},
  doi       = {10.1103/revmodphys.93.015005},
  publisher = {American Physical Society (APS)},
}

@Article{Liu2019,
  author    = {Liu, Chun-Hui and Chen, Shu},
  journal   = {Physical Review B},
  title     = {Topological classification of defects in non-Hermitian systems},
  year      = {2019},
  issn      = {2469-9969},
  month     = oct,
  number    = {14},
  pages     = {144106},
  volume    = {100},
  doi       = {10.1103/physrevb.100.144106},
  publisher = {American Physical Society (APS)},
}

@Article{Zeng2023,
  author    = {Zeng, Bowen and Yu, Tao},
  journal   = {Physical Review Research},
  title     = {Radiation-free and non-Hermitian topology inertial defect states of on-chip magnons},
  year      = {2023},
  issn      = {2643-1564},
  month     = jan,
  number    = {1},
  pages     = {013003},
  volume    = {5},
  doi       = {10.1103/physrevresearch.5.013003},
  publisher = {American Physical Society (APS)},
}

@Article{Lakkaraju2022,
  author    = {Lakkaraju, Leela Ganesh Chandra and Ghosh, Srijon and Sadhukhan, Debasis and Sen(De), Aditi},
  journal   = {Physical Review A},
  title     = {Mimicking quantum correlation of a long-range Hamiltonian by finite-range interactions},
  year      = {2022},
  issn      = {2469-9934},
  month     = nov,
  number    = {5},
  pages     = {052425},
  volume    = {106},
  doi       = {10.1103/physreva.106.052425},
  publisher = {American Physical Society (APS)},
}

@Article{Eldredge2017,
  author    = {Eldredge, Zachary and Gong, Zhe-Xuan and Young, Jeremy T. and Moosavian, Ali Hamed and Foss-Feig, Michael and Gorshkov, Alexey V.},
  journal   = {Physical Review Letters},
  title     = {Fast Quantum State Transfer and Entanglement Renormalization Using Long-Range Interactions},
  year      = {2017},
  issn      = {1079-7114},
  month     = oct,
  number    = {17},
  pages     = {170503},
  volume    = {119},
  doi       = {10.1103/physrevlett.119.170503},
  publisher = {American Physical Society (APS)},
}

@Article{Schauss2012,
  author    = {Schauß, Peter and Cheneau, Marc and Endres, Manuel and Fukuhara, Takeshi and Hild, Sebastian and Omran, Ahmed and Pohl, Thomas and Gross, Christian and Kuhr, Stefan and Bloch, Immanuel},
  journal   = {Nature},
  title     = {Observation of spatially ordered structures in a two-dimensional Rydberg gas},
  year      = {2012},
  issn      = {1476-4687},
  month     = oct,
  number    = {7422},
  pages     = {87--91},
  volume    = {491},
  doi       = {10.1038/nature11596},
  publisher = {Springer Science and Business Media LLC},
}

\normalsize

\section*{End Matter}\vspace{-3mm}

We now present additional non-Hermitian examples in \cref{fig:other_examples}, using \textit{LTP} to compute verified eigenvalues.

The first row of \cref{fig:other_examples} shows our method applied to an inverted quartic potential. Anharmonic oscillators are widely studied for testing techniques such as perturbation theory, Padé approximations, diagrammatic expansions, and variational or semiclassical methods. The considered model arises in quantum field theory and has been proposed to describe Higgs sector dynamics in the Standard Model \cite{Bender2001}. Here, boundary conditions are imposed not on the real line but on two Stokes sectors of angular width $\pi/3$, adjacent to and below the positive and negative real axes. We set $
x=-i\sqrt{1+2is}
$
for $s\in\mathbb{R}$ and discretize the transformed problem using Hermite functions and rectangular truncations of the resulting sparse infinite matrix. The eigenvalues are real and positive.

The second row of \cref{fig:other_examples} presents our method applied to an advection-diffusion operator on the periodic domain $[-\pi,\pi]$. This models the `explosive instability' of a rotating cylinder under gravity, neglecting inertial and capillary effects \cite{Benilov2003}, and belongs to a broader class of non-Hermitian fluid stability problems \cite{Fraternale2018}. We use a Fourier basis and apply rectangular truncations to the resulting sparse infinite matrix. The spectrum consists of real eigenvalues extending to $\pm\infty$.

The Papkovich--Fadle operator, arising in solid and fluid mechanics, concerns solutions $u$ of the biharmonic equation on a semi-infinite strip \cite{Papkovich1940,Fadle1940}. This operator attracts significant due to the non-normal nature of its eigenfunctions and expansions \cite{Shankar2007}. We approximate its eigenfunctions using a Chebyshev basis, orthonormalized to form an infinite basis representing the operator. Acting on block-2 functions over $[-1,1]$, the operator evolves solutions horizontally. Its matrix representation exhibits long-range interactions, so we work with matrices for $H$ and $H^*H$. Namely, letting $\sigma_{\mathrm{inf}}$ denote the smallest singular value, $\sigma_{\mathrm{inf}}((H-zI)\mathcal{P}_N)=\sqrt{\sigma_{\mathrm{inf}}(\mathcal{P}_N(H-zI)^*(H-zI)\mathcal{P}_N)}$. The eigenvalues occur in conjugate pairs with increasing imaginary parts.

Open systems often lead to non-Hermitian Hamiltonians due to the lack of energy conservation. Increasing the imaginary part of the potential beyond a critical threshold can trigger a transition from real to complex spectra (symmetry breaking). These phenomena are well studied and experimentally realized in optical systems \cite{regensburger2012parity,PhysRevLett.103.093902,ruter2010observation,feng2014single,hodaei2014parity,Makris2008}. We consider a lattice model on the discrete space $l^2(\mathbb{Z})$, incorporating an aperiodic complex potential superimposed on a harmonic potential. The aperiodicity arises from the incommensurability between the potential and the lattice. Our approach is versatile and capable of handling any type of potential.

In such systems, the spectra of truncated models are often highly sensitive to the imposed boundary conditions, exhibiting phenomena like the non-Hermitian skin effect \cite{Okuma2020}. The rectangular truncations employed here provide a natural set of boundary conditions that directly approximate the spectra of the infinite lattice model. This approach enables the algorithm to separate bulk states from edge states. 
In this example, some eigenvalues are real, while others form conjugate pairs. Increasing the value of $\alpha$ increases the imaginary parts of these eigenvalues. The fourth row of \cref{fig:other_examples} examines the broken $\mathcal{PT}$-symmetric regime.

\begin{figure}[t]
\centering
\includegraphics[width=0.48\textwidth]{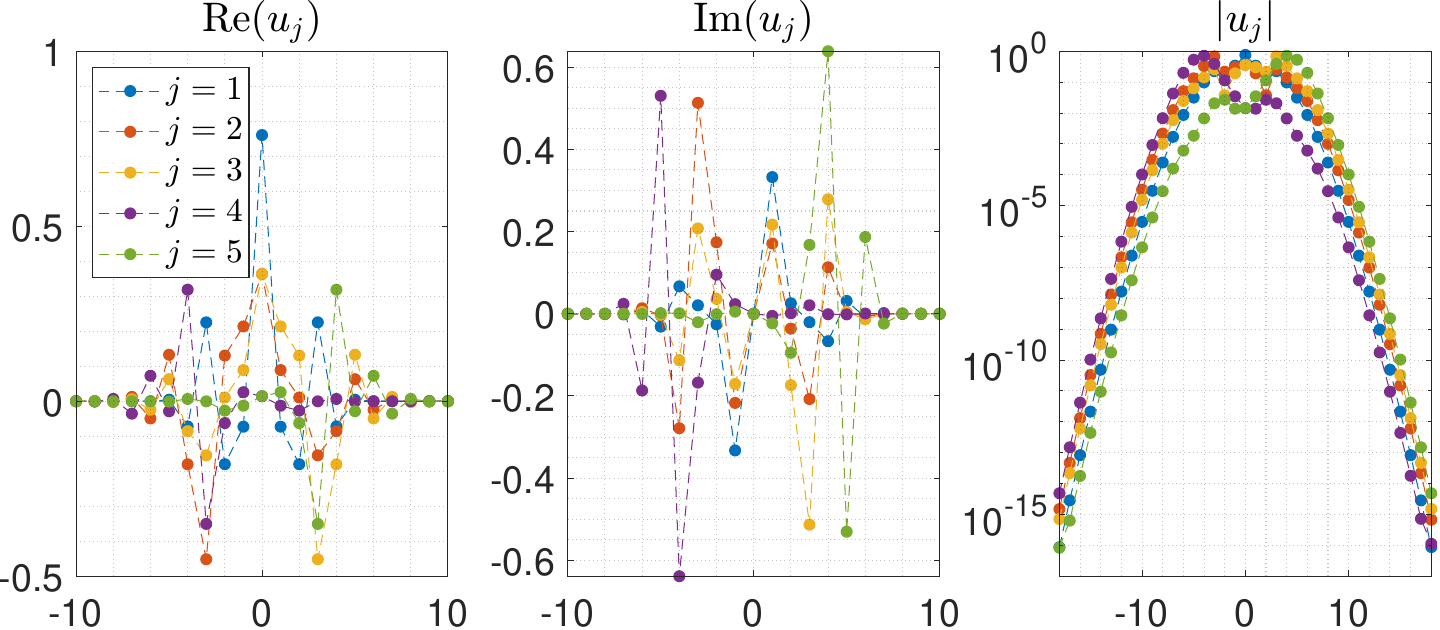}\vspace{-2mm}
\caption{First five eigenfunctions for the lattice operator (4th row of \cref{fig:other_examples}) computed with error bounded by $10^{-10}$.}\label{fig:lat_efuns}
\end{figure}

\begin{figure}[t]
\centering
\includegraphics[width=0.42\textwidth]{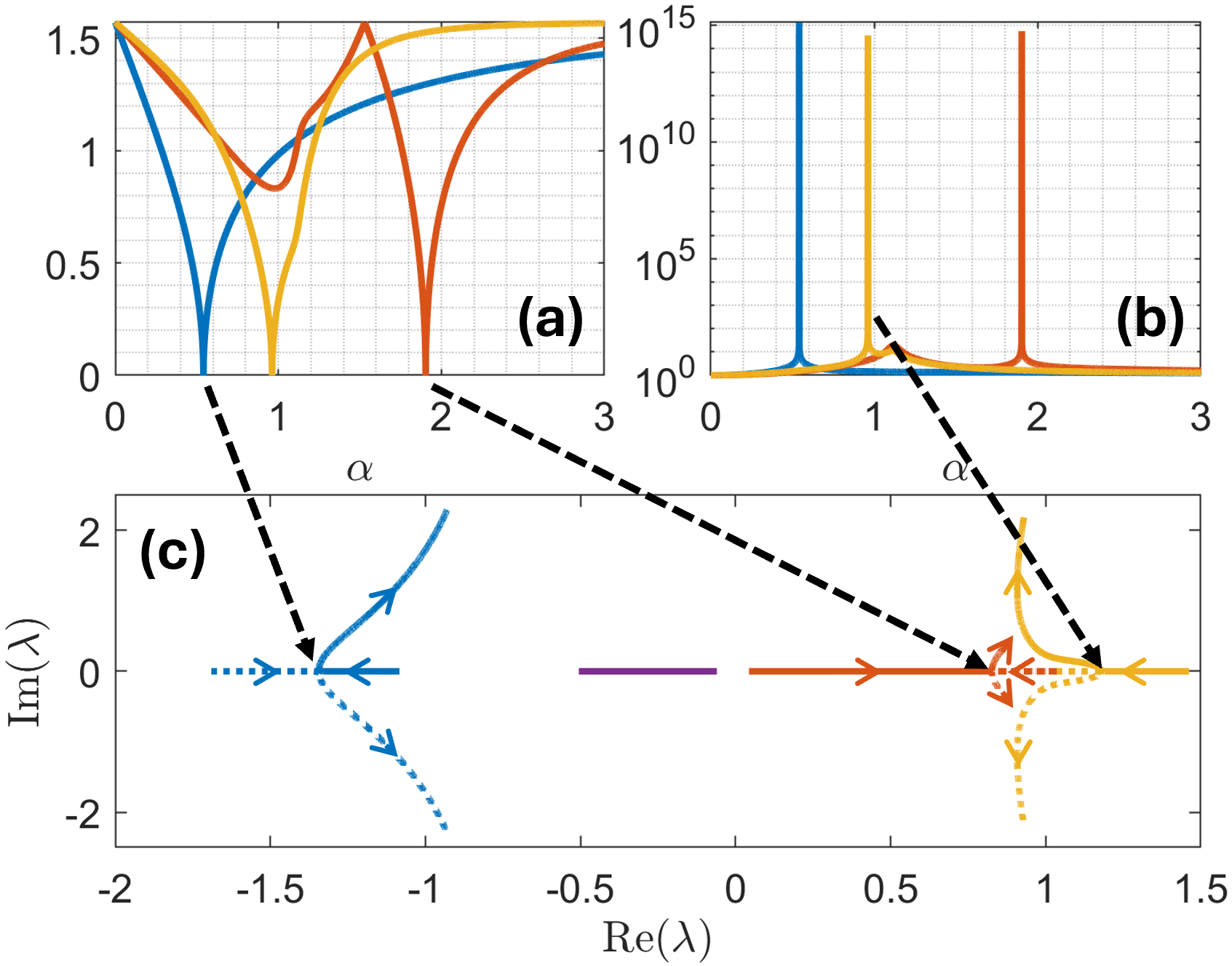}\vspace{-3mm}
\caption{Exceptional points for lattice operator (\cref{fig:other_examples}, 4th row). \textbf{(a)}: Subspace angle between eigenvectors computed with error bounded by $10^{-10}$. \textbf{(b)}: Blow-up of \textit{LTP} bound showing the need for higher-order \textit{LTP}. \textbf{(c)}: Trajectory of eigenvalues as $\alpha$ varies, with exceptional points shown by arrows. Dashed/solid lines show eigenvalue pairs that collide.}
\end{figure}

We can also compute eigenvectors with error bounds (see S.M. for details in the case of the imaginary cubic oscillator). \Cref{fig:lat_efuns} shows the first five eigenvectors computed for $\alpha = 2$. The $\mathcal{PT}$-symmetry enforces specific structures; for instance, the modes with $j = 2$ and $j = 3$ (corresponding to complex-conjugate eigenvalues) have real parts that are reflections of each other and imaginary parts that are negative reflections. As $\alpha$ increases from zero, exceptional points emerge when eigenvalues coalesce and become non-simple. At these points, the standard \textit{LTP} no longer holds. However, we can apply the generalized high-order version in \cref{generalizedLTP} with $p = 2$, corresponding to the algebraic multiplicity. This enables us to compute both eigenvalues and generalized eigenspaces with controlled error, as illustrated in \cref{fig:exeptional_points}. Notably, we can detect this phenomenon in the bulk, independent of boundary effects.

The final row of \cref{fig:other_examples} considers an operator modeling a laser resonator with Fresnel number $F$ and magnification $M$, tracing back to work by Siegman, Fox, Li, Landau, and others \cite{Siegman1986,Fox1963,Landau1976}. Laser operation involves two key eigenvalue problems: the Hermitian Schrödinger operator governing energy level differences (e.g., the neon frequency gap), and the optical cavity, which resonates at this frequency to produce coherence. The resulting non-Hermitian modes have attracted considerable interest \cite{Wright2020,Nixon2013}. Here, the eigenvalues are complex and spiral inward the accumulation point $0$.

\begin{figure}[t]
\label{fig:exeptional_points}\vspace{1mm}
\includegraphics[width=0.43\textwidth]{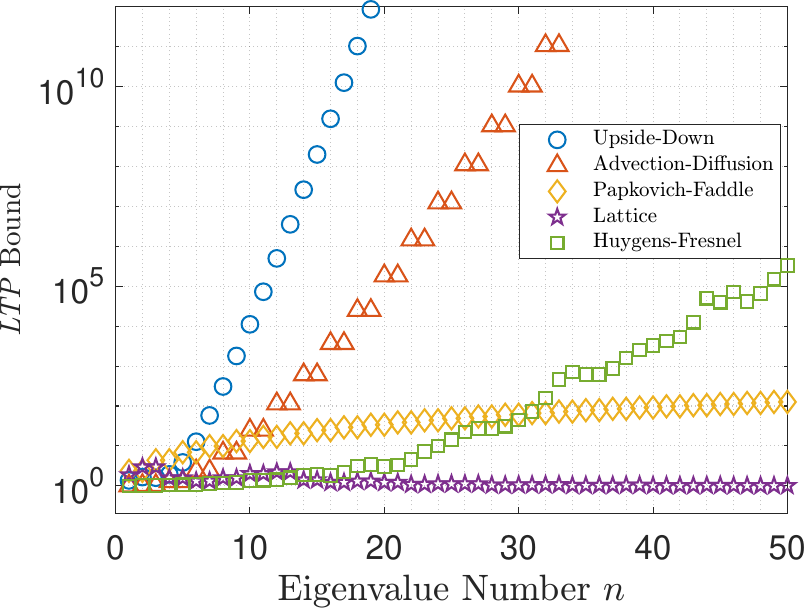}\vspace{-3mm}
\caption{The constants $C_K$ that locally bound the pseudospectra around each eigenvalue of the operators in \cref{fig:other_examples}.}\label{fig:conditions}
\end{figure}

To compute verified eigenvalues, \cref{fig:conditions} shows the \textit{LTP} bounds ($C_K$ constants) for the first 50 eigenvalues of each model, ordered by increasing absolute value (decreasing for the Huygens--Fresnel operator). For the upside-down potential, advection-diffusion, and Huygens--Fresnel operators, the bounds grow exponentially with index. For the Papkovich--Fadle operator, growth is algebraic. In the lattice model, the bounds converge to 1, as the unbounded real part of the potential dominates at high energies.

\begin{figure*}[t!]
\centering
\includegraphics[width=1\textwidth]{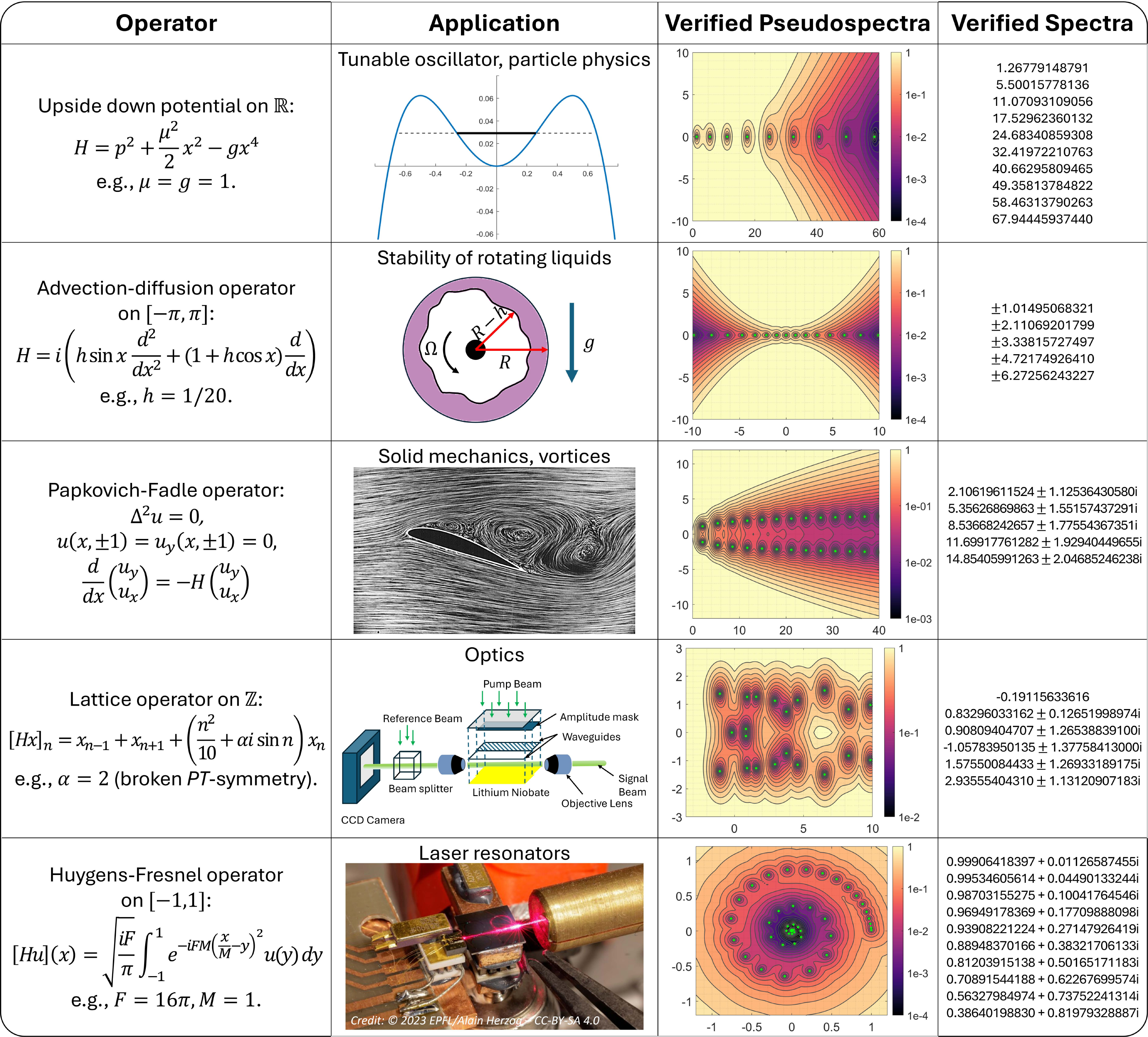}\vspace{-2mm}
\caption{Examples of other non-Hermitian operators for which we applied our method. In each case, we compute the first 10 non-trivial eigenvalues with verified absolute error $<10^{-10}$.}
\label{fig:other_examples}
\end{figure*}


\clearpage
\widetext

\begin{center}
\textbf{\large Supplemental Materials to\\``Computation and Verification of Spectra for Non-Hermitian Systems''\\\normalsize\vspace{4mm} Catherine Drysdale, Matthew Colbrook, Michael T. M. Woodley}
\end{center}

\setcounter{equation}{0}
\setcounter{figure}{0}
\setcounter{table}{0}
\setcounter{page}{1}
\makeatletter
\renewcommand{\theequation}{S\arabic{equation}}
\renewcommand{\thefigure}{S\arabic{figure}}
\renewcommand{\bibnumfmt}[1]{[S#1]}

\vspace{5mm}

We first discuss the two challenges, \textbf{(R1)} and \textbf{(R2)}, of computing spectra outlined in the main text. We then describe our method for calculating eigenvalues and eigenvectors of the imaginary cubic oscillator, which addresses these challenges. As outlined in our conclusion and demonstrated in our End Matter, the method is easily adaptable to other problems and generalizations.

\section{Fundamental limitations of computing spectra}

We now prove statements (A) and (B) from the main text. Before giving the proofs, we must define what an algorithm means. To make these impossibility results as strong as possible (which also makes the proofs simpler!), we use the setup of the Solvability Complexity Index (SCI) hierarchy \cite{Hansen_JAMS,ben2015can,colbrook2020PhD}. This allows us to classify the difficulty of problems using analysis techniques and examples from mathematical physics.

\subsection{Setting up the problem}

First, we will define a computational problem. Precision is needed since, for example, changing the information an algorithm can access dramatically affects the difficulty of a problem. If an algorithm can already access the eigenvalues of an operator, the problem becomes trivial. If, on the other hand, the information is too weak, the problem becomes too difficult. The following definition of a computational problem is deliberately general, designed to encompass all computational problems encountered in physics.

\begin{definition}[Computational problem]
\label{def:comp_prob}
The basic objects of a computational problem are:
\begin{itemize}[itemsep=0pt]
	\item A \textit{primary set}, $\Omega$, that describes the input class;
	\item A \textit{metric space} $(\mathcal{M},d)$;
	\item A \textit{problem function} $\Xi:\Omega\rightarrow\mathcal{M}$;
	\item An \textit{evaluation set}, $\Lambda$, of functions on $\Omega$.
\end{itemize}
The problem function $\Xi$ is the object we want to compute, with the notion of convergence captured by the metric space $(\mathcal{M},d)$. The evaluation set $\Lambda$ describes the information that algorithms can read. We refer to the collection $\{\Xi,\Omega,\mathcal{M},\Lambda\}$ as a \textit{computational problem}.
\end{definition}

\begin{example}
If $\Omega$ is a class of Hamiltonians, we could let $\Lambda$ be the evaluation of matrix coefficients with respect to some orthonormal basis. This is a typical setup throughout computations in quantum mechanics. We could consider the computation of the spectrum $\spec(H)$ for each $H\in\Omega$. If each $H$ is bounded, it is natural to let $(\mathcal{M},d)$ be the Hausdorff metric, which captures uniform convergence of compact subsets of $\mathbb{C}$:
$$
\dH(X,Y) = \max\left\{\sup_{x \in X} \inf_{y \in Y} |x-y|, \sup_{y \in Y} \inf_{x \in X} |x-y| \right\},\quad X,Y\in\MH.
$$
This means convergence without spurious eigenvalues or missing parts of the spectrum.
\end{example}

We can now define an algorithm, which is a function
$
\Gamma:\Omega\to \mathcal{M}
$
that, unlike the problem function $\Xi$, utilizes the evaluation set $\Lambda$ in some manner. For proving impossibility results, we use a deliberately general notion of algorithm.

\begin{definition}[General algorithm]
\label{def:Gen_alg}
Given a computational problem $\{\Xi,\Omega,\mathcal{M},\Lambda\}$, a {general algorithm} is a map $\Gamma:\Omega\to \mathcal{M}$ with the following property. For any $A\in\Omega$, there exists a non-empty finite subset of evaluations $\Lambda_\Gamma(A) \subset\Lambda$ such that if $B\in\Omega$ with $f(A)=f(B)$ for every $f\in\Lambda_\Gamma(A)$, then $\Lambda_\Gamma(A)=\Lambda_\Gamma(B)$ and $\Gamma(A)=\Gamma(B)$. In other words, 
the action of $\Gamma$ on $A$ can only depend on $\{f(A)\}_{f \in \Lambda_\Gamma(A)}$.
\end{definition}

\cref{def:Gen_alg} outlines the fundamental properties of any reasonable computational device. It states that $\Gamma$ can only use a finite amount of information, which it may select adaptively as it processes the input. The output of $\Gamma$ depends solely on the accessed information. Specifically, if $\Gamma$ encounters the same information for two different inputs, it must behave identically for both. A general algorithm has no restrictions on the operations allowed. It is more powerful than a digital or even analog computer and serves two main purposes. First, it significantly simplifies the process of proving lower bounds. The non-computability results we present are analytical and stem from the intrinsic non-computability of the problems themselves, not from the type of operations allowed being too restrictive. Second, the generality of \cref{def:Gen_alg} implies that a lower bound established for general algorithms also applies to any computational model. Moreover, our algorithms for computing spectra are executed using only arithmetic operations. Hence, we simultaneously derive the strongest possible lower and upper bounds on computational difficulty.

\begin{remark}
The SCI hierarchy is part of a broader program on classifying the difficulty of problems in mathematics and physics. We only need some of this hierarchy to prove the main text's results. The interested reader is pointed to \cite{colbrook2020PhD}, where various spectral problems from mathematical physics and other applications are classified.
\end{remark}

\subsection{The challenge of long-range interactions}

We now prove the first main result of the main text. The crux of the following proof (adapted from \cite{ben2015can}) is not knowing the corner entries in a large block, i.e.,  the lack of global information on long-range interactions in a Hamiltonian.

\begin{theorem}
\label{thm:SA_TWO_LIMITS_NEEDED}
Let $\OS$ denote the class of bounded self-adjoint operators acting on $l^2(\mathbb{N})$, $(\MH,\dH)$ be the collection of non-empty compact subsets of $\mathbb{R}$ equipped with the Hausdorff metric and $\Lambda$ be the evaluation of matrix entries (with respect to the canonical basis) of any $H\in\OS$. Then there is no sequence of general algorithms $\{\Gamma_n\}$ such that $\lim_{n\rightarrow\infty}\dH(\Gamma_n(H),\spec(H))=0$ for all $H\in\OS$.
\end{theorem}

\begin{proof}
Suppose for a contradiction that $\{\Gamma_{n}\}$ is a sequence of general tower of algorithms for $\{\spec,\OS,\MH,\Lambda\}$ such that $\lim_{n\rightarrow\infty}\dH(\Gamma_n(H),\spec(H))=0$ for all $H\in\OS$. Consider operators of the form
$$
H= \bigoplus_{r=1}^{\infty} H_{l_r},\quad l_r\in\mathbb{N}\backslash\{1\},\quad
H_k=\begin{pmatrix}
1& & & &1\\
 &0& & & \\
 & &\ddots& & \\
 & & &0& \\
1& & & &1\\
\end{pmatrix}
\in\mathbb{R}^{k\times k}.
$$
Then $H\in\OS$ with $\spec(H)=\spec(H_k)=\{0,2\}$. We will gain a contradiction by inductively selecting the integers $l_r$ so that $\Gamma_n(H)$ does not converge. For any $B\in\OS$ and $n\in\mathbb{N}$, define
$
N(B,n)=\max\{i,j : f_{i,j} \in \Lambda_{\Gamma_n}(B)\},
$
where $f_{i,j}(B)=\langle B e_j,e_i\rangle$. By \cref{def:Gen_alg}, we know that $N(B,n)$ is always finite. Moreover, if $B'\in\OS$ is such that $f_{i,j}(B')=f_{i,j}(B)$ for all $i,j\leq N(B,n)$, then $\Gamma_n(B)=\Gamma_n(B')$.

For $C=\diag(1,0,0,\ldots)\in\OS$, $\spec(C)=\{0,1\}$. It follows from the assumed convergence that there exists $n_0\in\mathbb{N}$ such that $\dist(1,\Gamma_{n_0}(C))<1/2$. If $B'\in\OS$ is such that $f_{i,j}(B')=f_{i,j}(B)$ for all $i,j\leq N(B,n)$, then $\Gamma_n(B)=\Gamma_n(B')$. We choose $l_1\geq 2$ so that $f_{i,j}(H)=f_{i,j}(C)$ for any $i,j\leq N(C,n_0)$ and hence $\Gamma_{n_0}(H)=\Gamma_{n_0}(C)$. In particular, this implies that $\dist(1,\Gamma_{n_0}(H))<1/2$. We now proceed inductively. Suppose that $l_1,\ldots,l_k$ have been chosen and let $C_k=H_{l_1} \oplus \cdots \oplus  H_{l_k} \oplus C\in\OS$. Clearly, $1\in\spec(C_k)$ and hence there exists $n_{k}>n_{k-1}$ so that $\dist(1,\Gamma_{n_{k}}(C_k))<1/2$. Arguing as before, we may select $l_{k+1}\geq 2$ so that $f_{i,j}(H)=f_{i,j}(C_k)$ for any $i,j\leq N(C_k,n_k)$ and hence so that $\Gamma_{n_k}(H)=\Gamma_{n_k}(C)$. In particular, this implies that $\dist(1,\Gamma_{n_k}(H))<1/2$. This holds for all $k\in\mathbb{N}$. But this contradicts $\lim_{n\rightarrow\infty}\Gamma_n(H)=\{0,2\}$.
\end{proof}

\begin{remark}[Extension to approximate eigenvectors]
The proof immediately extends to computing approximate eigenvectors. The only change is that instead of approximating the spectral point $1$, the supposed algorithms approximate the corresponding eigenvector of $C$.
\end{remark}

\subsection{The challenge of non-Hermiticity}

We now consider the class of tridiagonal (possibly non-Hermitian) operators. The following example first examines the (in)stability properties of eigenvalues of finite matrices as their dimensions increase.

\begin{example}[Instability of spectrum of Jordan blocks]
\label{exam:instab_jordan}
Consider the $k\times k$ Jordan block and its resolvent:
\begin{equation}
\label{eq:jordan_block}
J_k=
\begin{pmatrix}
0& 1& & \\
  &\ddots&\ddots & \\
  & &0& 1\\
 & & &0\\
\end{pmatrix},\quad (J_k-zI)^{-1}=
\begin{pmatrix}
\frac{-1}{z}& \frac{-1}{z^2} &\hdots &\frac{-1}{z^k}\\
  &\ddots&\ddots & \vdots\\
  & &\frac{-1}{z}& \frac{-1}{z^2}\\
 & & &\frac{-1}{z}\\
\end{pmatrix},
\end{equation}
where $(J_k-zI)^{-1}$ exists for $z\neq 0$. If we fix $z\in\mathbb{C}$ with $0<|z|<1$, then
$$
\|(J_k-zI)^{-1}\|\geq |\langle(J_k-zI)^{-1}e_k,e_1\rangle|=|z|^{-k}.
$$
Hence, the resolvent norm increases exponentially with the dimension $k$. Due to the characterization of pseudospectra as a union of spectra of perturbed operators, this exponential increase indicates a severe instability of the spectrum.
\end{example}

In infinite dimensions, this instability is used to show that computing spectra of infinite tridiagonal operators is impossible with a sequence of algorithms. The proof of the following theorem is adapted from \cite{ben2015can}.

\begin{theorem}
\label{thm:BANDED_TWO_LIMITS_NEEDED}
Let $\OT$ denote the class of bounded operators acting on $l^2(\mathbb{N})$ that are tridiagonal with respect to the canonical basis. Let $(\MH,\dH)$ be the collection of non-empty compact subsets of $\mathbb{C}$ equipped with the Hausdorff metric. Let $\Lambda$ be the evaluation of matrix entries, with respect to the canonical basis, of any $H\in\OT$. Then there is no sequence of general algorithms $\{\Gamma_n\}$ such that $\lim_{n\rightarrow\infty}\dH(\Gamma_n(H),\spec(H))=0$ for all $H\in\OT$.
\end{theorem}

\begin{proof}
Suppose for a contradiction that $\{\Gamma_{n}\}$ is a sequence of general tower of algorithms for $\{\spec,\OT,\MH,\Lambda\}$ such that $\lim_{n\rightarrow\infty}\dH(\Gamma_n(H),\spec(H))=0$ for all $H\in\OT$. Consider operators of the form
$$
H= \bigoplus_{r=1}^{\infty} J_{l_r},\quad l_r\in\mathbb{N}\backslash\{1\},
$$
where the $J_{l_r}$ are Jordan blocks defined in \cref{eq:jordan_block}. If the sequence $\{l_r\}\subset\mathbb{N}\backslash\{1\}$ is bounded, then $\spec(H)=\{0\}$. Otherwise, owing to the fact that $\lim_{k\rightarrow\infty}\|(J_k-zI)^{-1}\|=\infty$ for $|z|\leq 1$, we see that $\spec(H)=\{z\in\mathbb{C}: |z|\leq1\}$. We will gain a contradiction by choosing the integers $l_r$ so that $\Gamma_n(H)$ does not converge.

We take $N(\cdot,n)$ as in the proof of \cref{thm:SA_TWO_LIMITS_NEEDED}. Set $B_m=\bigoplus_{r=1}^{\infty} J_{m}\in\OT$ for any $m\in\mathbb{N}$. Note that $\spec(B_m)=\{0\}$ for any $m$. Hence there exists $n_0\in\mathbb{N}$ such that $\dist(1,\Gamma_{n_0}(B_2))>1/2$. We may pick $l_r=2$ for $r=1,\ldots, N_1$ and some $N_1$ so that $f_{i,j}(H)=f_{i,j}(B_2)$ for any $i,j\leq N(B_2,n_0)$ and hence $\Gamma_{n_0}(H)=\Gamma_{n_0}(B_2)$. In particular, this implies that $\dist(1,\Gamma_{n_0}(H))>1/2$. We now proceed inductively. Suppose that $l_r$ have been chosen for $r=1,\ldots, N_k$ and let $C_k=J_{l_1} \oplus \cdots \oplus  J_{l_{N_k}} \oplus B_{k+2}\in\OT$. Clearly, $\spec(C_k)=\{0\}$ and hence there exists $n_{k}>n_{k-1}$ so that $\dist(1,\Gamma_{n_{k}}(C_k))>1/2$. Arguing as before, we may select $l_r=k+2$ for $r=N_k+1,\ldots, N_{k+1}$ and some $N_{k+1}$ so that $f_{i,j}(H)=f_{i,j}(C_k)$ for any $i,j\leq N(C_k,n_k)$ and hence so that $\Gamma_{n_k}(H)=\Gamma_{n_k}(C_k)$. This implies that $\dist(1,\Gamma_{n_k}(H))>1/2$. Hence $\dist(1,\Gamma_{n_k}(H))>1/2$ for all $k\in\mathbb{N}$. But this contradicts $\lim_{n\rightarrow\infty}\Gamma_n(H)=\{z\in\mathbb{C}: |z|\leq1\}$.
\end{proof}

\begin{remark}The impossibility result of \cref{thm:BANDED_TWO_LIMITS_NEEDED} is a stronger statement than numerical instability. It holds even if we permit exact arithmetic without any round-off errors. The proof also holds when we replace the Jordan blocks with any family of tridiagonal finite matrices whose resolvent blows up as the dimension increases.
\end{remark}

\begin{remark}[Extension to approximate eigenvectors]
The proof immediately extends to computing approximate eigenvectors. However, now we use the fact that $1$ is in the spectrum of $H$, but the supposed algorithms do not approximate associated approximate eigenvectors.
\end{remark}

\section{Verified eigenvalues and eigenfunctions of the imaginary cubic oscillator}

\subsection{Rectangular truncations and convergence to eigenpairs}

To compute spectral properties of $H_{\mathrm{B}}$, we use Hermite functions $
u_m(x)=e^{-x^2/2}H_{m}(x)/{\sqrt{2^{m}m!\sqrt{\pi}}}, m\in\mathbb{Z}_{\geq 0},
$
where
$
H_m(x)=(-1)^me^{x^2}{\mathrm{d}^m}e^{-x^2}/{\mathrm{d}x^m}.
$
Recall that $\mathcal{P}_N$ denotes the projection onto $\mathrm{span}\{u_0,\ldots,u_{N-1}\}$ and that the matrix representation of $H_{\mathrm{B}}$ with respect to Hermite functions is banded. This follows from the recursion relations for derivatives and multiplication by $x$ of Hermite functions. Explicitly, we have:
\begin{align*}
[H_{\mathrm{B}}u_m](x)&=-\frac{\sqrt{m(m-1)}}{2}u_{m-2}(x) +\frac{m^2+(m+1)^2}{2}u_m(x)-\frac{\sqrt{(m+1)(m+2)}}{2}u_{m+2}(x)\\
&\quad+i\left[\frac{\sqrt{m(m-1)(m-2)}}{2\sqrt{2}}u_{m-3}(x)+\left(\frac{(m-1)\sqrt{m}}{2\sqrt{2}}+\frac{m^2+(m+1)^2}{2}\sqrt{\frac{m}{2}}\right)u_{m-1}(x)\right]\\
&\quad\quad +i\left[\frac{\sqrt{(m+1)(m+2)(m+3)}}{2\sqrt{2}}u_{m+3}(x)+\left(\frac{(m+2)\sqrt{m+1}}{2\sqrt{2}}+\frac{m^2+(m+1)^2}{2}\sqrt{\frac{m+1}{2}}\right)u_{m+1}(x)\right].
\end{align*}
For an operator $T$ with domain $\mathcal{D}(T)$, we let
$$
\sigma_{\mathrm{inf}}(T)=\inf\{\|Tx\|:x\in\mathcal{D}(T),\|x\|=1\}.
$$
This quantity is known as the \textit{injection modulus} and is the operator generalization of the smallest singular value. It is convenient to use the notation $\sigma_{\mathrm{inf}}$ since it applies to both finite and infinite-dimensional operators. Our first result shows that we can compute the function $z\mapsto\|(H_{\mathrm{B}}-zI)^{-1}\|^{-1}$ by computing smallest singular values of \textit{finite rectangular} matrices. In other words, the eigenfunctions of the harmonic oscillator provide a suitable basis for computing the spectral properties of the imaginary cubic oscillator.

\begin{proposition}[Rectangular truncations to compute resolvent norm]
\label{gamma_n_conv}
Let $\gamma_N(z)=\sigma_{\mathrm{inf}}(\mathcal{P}_{N+3}(H_{\mathrm{B}}-zI)\mathcal{P}_N)$. Then $\gamma_N(z)\geq \|(H_{\mathrm{B}}-zI)^{-1}\|^{-1}$ and $\gamma_N(z)$ converges to $\|(H_{\mathrm{B}}-zI)^{-1}\|^{-1}$ uniformly on compact subsets of $\mathbb{C}$ as $N\rightarrow\infty$.
\end{proposition}

\begin{proof}
Since compactly supported smooth functions form a core of $H_{\mathrm{B}}$, the linear span of Hermite functions forms a core for $H_{\mathrm{B}}$ \cite[Proposition 7.1]{colbrook3}. It follows that \cite[Theorem 6.7]{colbrook3}
$
\lim_{N\rightarrow\infty}\sigma_{\mathrm{inf}}((H_{\mathrm{B}}-zI)\mathcal{P}_N)=\sigma_{\mathrm{inf}}(H_{\mathrm{B}}-zI),
$
with convergence from above, which is uniform on compact subsets of $\mathbb{C}$. Since the spectrum of $H_{\mathrm{B}}$ is discrete, $\sigma_{\mathrm{inf}}(H_{\mathrm{B}}-zI)=\sigma_{\mathrm{inf}}(H_{\mathrm{B}}^*-\overline{z}I)$ and hence \cite[Lemma 6.4]{colbrook3}
$$
\|(H_{\mathrm{B}}-zI)^{-1}\|^{-1}=\min\{\sigma_{\mathrm{inf}}(H_{\mathrm{B}}-zI),\sigma_{\mathrm{inf}}(H_{\mathrm{B}}^*-\overline{z}I)\}=\sigma_{\mathrm{inf}}(H_{\mathrm{B}}-zI).
$$
The statement of the proposition follows by noting that $\gamma_N(z)=\sigma_{\mathrm{inf}}((H_{\mathrm{B}}-zI)\mathcal{P}_N)$ from the bandedness of the matrix representation of $H_{\mathrm{B}}$ using Hermite functions.
\end{proof}

There is a rich literature on using injection moduli to compute spectra of Hermitian operators \cite{kato1949upperH,zimmermann1995variational,barrenechea2014finite} going back to the work of Kato. The most general form of convergence is achieved by the algorithm \texttt{CompSpec} of Colbrook et al. \cite{PhysRevLett.122.250201,colbrook3}, which converges for any Hermitian operator (with no assumptions on essential spectra), even those with complicated fractal spectra. In this paper, the fundamental challenge is the non-Hermiticity of the operator $H_{\mathrm{B}}$.

We now show how finding \textit{local minimizers} of $\gamma_N$ allows us to compute the eigenfunctions and eigenvalues of $H_{\mathrm{B}}$. Suppose an interval $[a,b]\subset\mathbb{R}$ contains a desired eigenvalue $\lambda_n$ (and no other eigenvalues). In the case of $H_{\mathrm{B}}$, we can find such intervals using asymptotic formulas. Let $z_N$ be a global minimizer of the continuous function $\gamma_N$ on this interval, and let $f_N$ denote the corresponding right-singular vector of $\mathcal{P}_{N+3}(H_{\mathrm{B}}-z_NI)\mathcal{P}_N$ corresponding to the smallest singular value. Let $E_n$ denote the eigenspace $\mathrm{span}\{\phi_n\}$ and define the angle between two subspaces $\mathrm{span}\{u\},\mathrm{span}\{w\}$ as 
$$
\angle(\mathrm{span}\{u\},\mathrm{span}\{w\})=\min\left\{\cos^{-1}\left(\frac{\langle v_1,v_2\rangle}{\|v_1\|\|v_2\|}\right):v_1\in\mathrm{span}\{u\},v_2\in\mathrm{span}\{w\}\right\}.
$$
The subspace angle measures the distance (or angle) between subspaces and is the standard metric in measuring convergence to eigenspaces. The following theorem shows that this minimization process (to find $z_N$ and $f_N$) converges to $\lambda_n$ and the eigenspace $E_n$ as $N\rightarrow\infty$. In particular, the monotonicity in \cref{gamma_n_conv} ensures that we converge without spurious eigenvalues. This monotonicity does not hold for square truncations of operators -- it is the rectangular truncation that makes this possible.

\begin{theorem}[Convergence to eigenpairs]\label{SM_conv_thm}
Given the above, $\lim_{N\rightarrow\infty}z_N=\lambda_n$ and $\lim_{N\rightarrow\infty}\angle(\mathrm{span}\{f_N\},E_n)=0$.
\end{theorem}

\begin{proof}
Since $\|(H_{\mathrm{B}}-\lambda_nI)^{-1}\|^{-1}=0$, the convergence in \cref{gamma_n_conv} and definition of $z_N$ as a minimizer implies that $\lim_{N\rightarrow\infty}\gamma_N(z_N)=0$. Since $\gamma_{N}(z)\geq \|(H_{\mathrm{B}}-zI)^{-1}\|^{-1}$, it follows that $z_N\in\spec_{\epsilon_N}(H_{\mathrm{B}})$ for a sequence $\{\epsilon_N\}$ with $\lim_{N\rightarrow\infty}\epsilon_N=0$ and hence $\lim_{N\rightarrow\infty}\dist(z_N,\spec(H_{\mathrm{B}}))=0$. Since $[a,b]\cap\spec(H_{\mathrm{B}})=\{\lambda_n\}$, $\lim_{N\rightarrow\infty}z_N=\lambda_n$. Suppose for a contradiction that $\lim_{N\rightarrow\infty}\angle(\mathrm{span}\{f_N\},E_n)\neq 0$. Then, by picking a subsequence if necessary, without loss of generality
$
f_N= a_N+b_N,
$
with $a_N\in E_n$, $b_N\in E_n^\perp$ and $\|b_N\|\geq\delta>0$. Hence,
$$
\gamma_N(z_N)=\|(H_{\mathrm{B}}-z_NI)(a_N+b_N)\|\geq \|(H_{\mathrm{B}}-z_NI)b_N\|-\|(H_{\mathrm{B}}-z_NI)a_N\|=
\|(H_{\mathrm{B}}-z_NI)b_N\|-|\lambda_n-z_N|\|a_N\|.
$$
The second term on the right-hand side converges to zero, while the first term is asymptotically bounded below by $\sigma_{\mathrm{inf}}((H_{\mathrm{B}}-\lambda_nI)|_{E_n^\perp})\delta>0$. But this contradicts  $\lim_{N\rightarrow\infty}\gamma_N(z_N)=0$.
\end{proof}

We use a two-stage process to obtain verified values of $\gamma_N$. First, we compute candidate pairs $(z_N,f_N)$ using floating-point arithmetic. To find the minimizer over an interval, we use a bisection method (also called the golden search) \cite[Section 10.2]{press2007numerical}, which rapidly converges with a few evaluations of $\gamma_N$. We perform this step using extended precision if we want eigenvalues to a huge number of digits (such as in the table in the main text). With this candidate in hand, we obtain an upper bound on $\|(H_{\mathrm{B}}-z_NI)f_N\|$ and hence $\gamma_N$ using interval arithmetic. This step is straightforward and efficient, owing to the explicit action of $H_{\mathrm{B}}$ on Hermite functions. For example, \cref{fig:residuals} shows the convergence to the fifth eigenvalue, $\lambda_5$. We have shown the values of $\gamma_N$ and the subspace angle (computed by comparing to a converged eigenfunction) is single, double, and quadruple precision. The convergence is exponential, and this is the case for all of the eigenvalues of $H_{\mathrm{B}}$. To convert the bounds on $\gamma_N$ into bounds on eigenvalues, we must control the resolvent of $H_{\mathrm{B}}$. We turn to this problem in the next section.

\begin{figure}[t]
\centering
\includegraphics[width=0.46\textwidth,trim=0mm 0mm 0mm 0mm,clip]{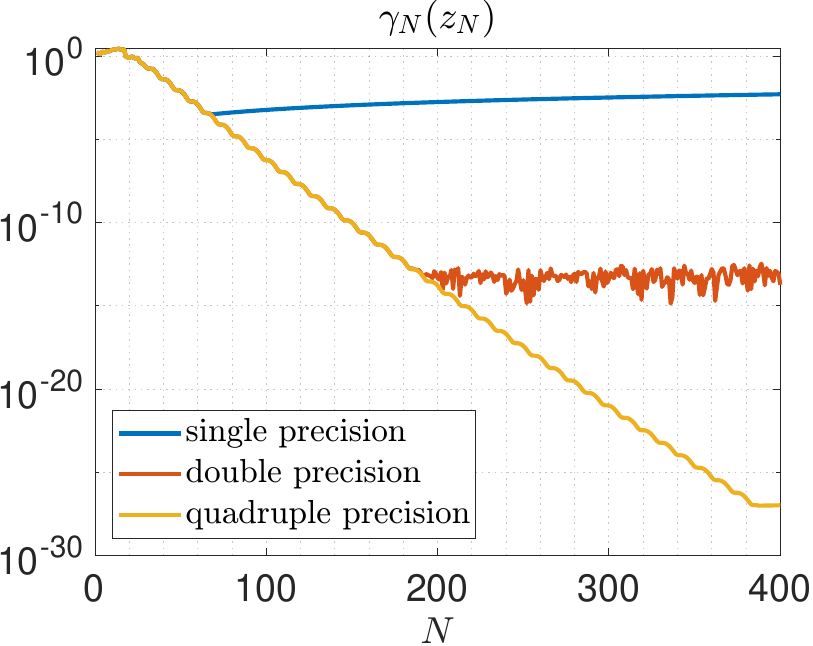}
\hfill
\includegraphics[width=0.46\textwidth,trim=0mm 0mm 0mm 0mm,clip]{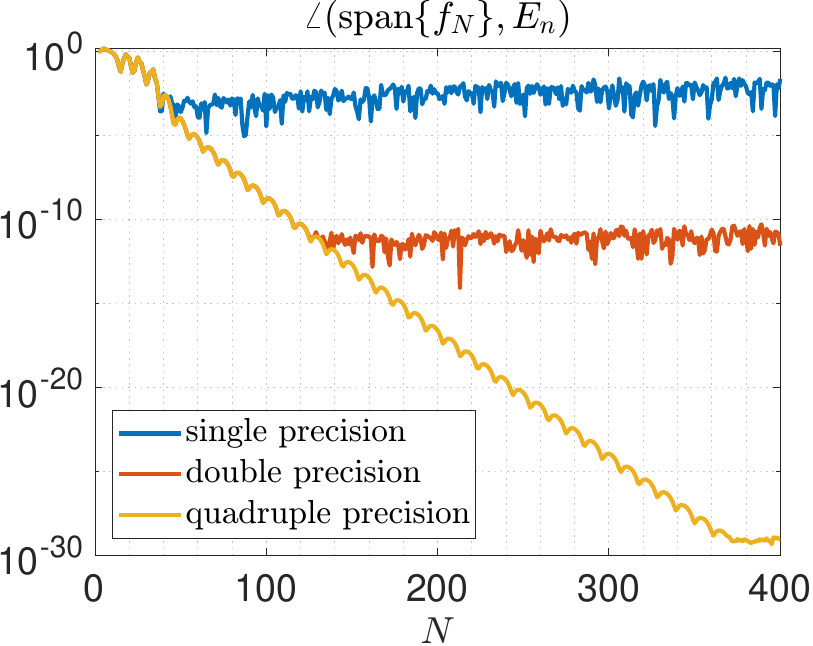}
\caption{Illustration of \cref{SM_conv_thm} for the fifth eigenvalue, $\lambda_5$. The method converges exponentially to eigenvalues and eigenfunctions. Left: Exponential convergence of the residuals $\gamma_N$ to zero. Right: Exponential convergence to the eigenfunction. These bounds are then verified using interval arithmetic.}
\label{fig:residuals}
\end{figure}

\subsection{Resolvent bound and verification (and proof of Theorem 3 from main text)}

In this section, we prove the result in the main text that bounds the resolvent norm of the imaginary cubic oscillator. Recall that $H_{\mathrm{B}}$ denotes the imaginary cubic oscillator. Part of our proof adapts and generalizes some bounds in \cite{dondl2017bound} that showed that the unbounded part of the pseudospectrum of $H_{\mathrm{B}}$ escapes towards $+\infty$ as $\epsilon$ decreases.

\begin{proof}[Proof of resolvent bound in main text]
The operator
$$
\mathcal{Q}_n=\frac{-1}{2\pi i}\int_{\tau_n} (H_{\mathrm{B}}-zI)^{-1}\dd z
$$
is the Riesz projection corresponding to the $n$th eigenvalue, $\lambda_n$, of $H_{\mathrm{B}}$ \cite[p. 178]{kato2013perturbation}. Here, $\tau_n$ is a contour that wraps once around the eigenvalue $\lambda_n$ and no other eigenvalues. The connection here is that the condition number from the main text is given by $\kappa_n=\|\mathcal{Q}_n\|$.
The operator $H_{\mathrm{B}}$ has compact resolvent, its eigenfunctions form a complete family, and its eigenvalues have algebraic multiplicity one \cite{dorey2001spectral,tai2006simpleness}. The space $L^2(\mathbb{R})$ decomposes into a (non-orthogonal) direct sum of closed, $H_{\mathrm{B}}$-invariant subspaces
$$
L^2(\mathbb{R})=\mathrm{Ran}(\mathcal{Q}_1)\oplus \cdots \oplus \mathrm{Ran}(\mathcal{Q}_m)\oplus \mathrm{Ran}(I-\mathcal{S}_m),
$$
where $\mathcal{S}_m=\sum_{n=1}^m\mathcal{Q}_n$. Using the fact that each $\mathcal{Q}_n$ commutes with $H_{\mathrm{B}}$, we have
$$
(H_{\mathrm{B}}-zI)^{-1}=( I- \mathcal{S}_m)(H_{\mathrm{B}}-zI)^{-1}( I- \mathcal{S}_m)+\sum_{n=1}^m(\lambda_n-z)^{-1}\mathcal{Q}_n.
$$
It follows that
\begin{align}
\|(H_{\mathrm{B}}- zI)^{-1} \|\leq \left(1+\sum_{n=1}^m\|\mathcal{Q}_n\|\right)\|(H_{\mathrm{B}}-zI)|_{\mathrm{Ran}(I-\mathcal{S}_m)}^{-1}\|+\sum_{n=1}^m\frac{\|\mathcal{Q}_n\|}{|z-\lambda_n|}.\label{eq:cubic_ex1}
\end{align}
To bound  $\|(H_{\mathrm{B}}-zI)|_{\mathrm{Ran}(I-\mathcal{S}_m)}^{-1}\|$ we will use the Hille--Yosida theorem \cite{pazy2012semigroups}. We can do this since $-H_{\mathrm{B}}$ and $-H_{\mathrm{B}}|_{\mathrm{Ran}(I-\mathcal{S}_m)}$ generate contraction semigroups $\{S(t)\}_{t\geq 0}$ with $\|S(t)\|\leq1$ (the operators $H_{\mathrm{B}}$ and $H_{\mathrm{B}}|_{\mathrm{Ran}(I-\mathcal{S}_m)}$ are $m$-accretive). We will also use the following lemma from \cite[Lemma 3.1]{davies2005semigroup}.
\begin{lemma}
\label{lem:davies_lemma}
Let $\{S(t)\}_{t\geq 0}$ be a strongly continuous semigroup on a Banach space and $\{\psi_n\}_{n=1}^\infty$
be a complete set of linearly independent vectors. Let $S_n(t)$ denote the restriction of $S(t)$ to $\mathrm{span}\{\psi_1,\ldots,\psi_n\}$. Then for any $t\geq 0$,
$
\|S(t)\|=\lim_{n\rightarrow\infty}\|S_n(t)\|.
$
\end{lemma}

Using \cref{lem:davies_lemma}, we observe that
$$
\| e^{-t H_{\mathrm{B}}}|_{\mathrm{Ran}(I-\mathcal{S}_m)}\| \leq e^{-t \lambda_{m+1}} \sum_{n=m+1}^{\infty} e^{-t(\lambda_n - \lambda_{m+1})}\|\mathcal{Q}_n\|
$$
We assume that $\| Q_n \| \leq \exp(\frac{n\pi}{\sqrt{3}})$ and use the fact that $\lambda_n - \lambda_{m+1} \geq (n - (m+1)) (\frac{\pi}{\sqrt{3}} + 1)$. We have for $t \geq 1$,
\begin{align*}
\sum_{n={m+1}}^\infty e^{-t(\lambda_n-\lambda_{m+1})}{\|\mathcal{Q}_n\|}&\leq \sum_{n={m+1}}^\infty e^{-(\lambda_n-\lambda_{m+1})}{\|\mathcal{Q}_n\|}\\
&\leq \sum_{n={m+1}}^\infty e^{-(n-(m+1))({\pi}/{\sqrt{3}}+1)}\exp(n\pi/\sqrt{3})\\
&= \exp\left({(m+1)\pi}/{\sqrt{3}}\right) \sum_{n={m+1}}^\infty e^{m+1-n}= \frac{e}{e-1}\exp\left({(m+1)\pi}/{\sqrt{3}}\right).
\end{align*}
We must also consider the case that $t<1$. Here, we use the facts that $\|e^{-tH_{\mathrm{B}}}|_{\mathrm{Ran}(I-\mathcal{S}_m)}\|\leq 1$ for any $t$ and $\lambda_{m+1}\leq c(m+1)^{6/5}$ with $c=[2 \Gamma(\frac{11}{6})\sqrt{\pi}/({\sqrt{3} \Gamma(\frac{4}{3})})]^{6/5}$ to see that if $t<1$, then
$$
\|e^{-tH_{\mathrm{B}}}|_{\mathrm{Ran}(I-\mathcal{S}_m)}\|\leq e^{-\lambda_{m+1}}e^{c(m+1)^{6/5}} \leq e^{-t\lambda_{m+1}}e^{c(m+1)^{6/5}}.
$$
It follows that
$$
\|e^{-tH_{\mathrm{B}}}|_{\mathrm{Ran}(I-\mathcal{S}_m)}\|\leq e^{-t\lambda_{m+1}}\max\left\{\frac{e}{e-1}\exp\left({(m+1)\pi}/{\sqrt{3}}\right),e^{c(m+1)^{6/5}}\right\}\leq e^{-t\lambda_{m+1}}e^{c(m+1)^{6/5}}.
$$
The Hille--Yosida theorem now implies that
$$
\|(H_{\mathrm{B}}-zI)|_{\mathrm{Ran}(I-\mathcal{S}_m)}^{-1}\|\leq \frac{e^{c(m+1)^{6/5}}}{\lambda_{m+1}-\mathrm{Re}(z)}\quad \forall z\in\mathbb{C}\text{ with }\mathrm{Re}(z)<\lambda_{m+1}.
$$
Using $\| Q_n \| \leq \exp(\frac{n\pi}{\sqrt{3}})$ again, we see that if $\mathrm{Re}(z)\leq \lambda_m$, then
$$
\left(1+\sum_{n=1}^m\|\mathcal{Q}_n\|\right)\|(H_{\mathrm{B}}-zI)|_{\mathrm{Ran}(I-\mathcal{S}_m)}^{-1}\|\leq \frac{e^{(m+1)\pi/\sqrt{3}}-1}{e^{\pi/\sqrt{3}}-1}\cdot\frac{e^{c(m+1)^{6/5}}}{\pi/\sqrt{3}+1}.
$$
Combining with \cref{eq:cubic_ex1}, we see that if $\lambda_{m-1}\leq \mathrm{Re}(z)\leq \lambda_m$ with $z\not\in\spec(H_{\mathrm{B}})$, then
\begin{align*}
\|(H_{\mathrm{B}}-zI)^{-1}\|&\leq 
\frac{\|\mathcal{Q}_{m-1}\|}{|\lambda_{m-1}-z|}+\frac{\|\mathcal{Q}_{m}\|}{|\lambda_{m}-z|}+\frac{1}{\pi/\sqrt{3}+1}\left[
\frac{e^{(m-1)\pi/\sqrt{3}}-1}{e^{\pi/\sqrt{3}}-1}-1+
\frac{e^{(m+1)\pi/\sqrt{3}}-1}{e^{\pi/\sqrt{3}}-1}\cdot e^{c(m+1)^{6/5}}\right]\\
&\leq \frac{\|\mathcal{Q}_{m-1}\|}{|\lambda_{m-1}-z|}+\frac{\|\mathcal{Q}_{m}\|}{|\lambda_{m}-z|}
+ \frac{\exp[(m+1)\pi/\sqrt{3}+c(m+1)^{6/5}]}{14}.
\end{align*}
This completes the proof of the bound in the main text.
\end{proof}

Using this resolvent bound and combining with \cref{gamma_n_conv}, we have that for any $N\in\mathbb{N}$ and $\lambda_{m-1}< \mathrm{Re}(z_N)< \lambda_m$ with $c_m\gamma_N(z_N)<1$,
$$
\mathrm{dist}(z_N,\mathrm{Sp}(H_{\mathrm{B}}))\leq \frac{2\exp(m\frac{\pi}{\sqrt{3}})\|(H_{\mathrm{B}}-z_NI)^{-1}\|^{-1}}{1-c_m\|(H_{\mathrm{B}}-z_NI)^{-1}\|^{-1}}\leq \frac{2\exp(m\frac{\pi}{\sqrt{3}})\gamma_N(z_N)}{1-c_m\gamma_N(z_N)}.
$$
This means that we can explicitly bound the error of the eigenvalues computing using our minimization procedure in the previous subsection. Moreover, this error bound converges to zero as $N\rightarrow\infty$ since $\lim_{N\rightarrow\infty}\gamma_{N}(z_N)=0$. Since we are bounding the projection norms, a similar bound allows us to bound the subspace angles when computing the eigenfunctions.

\subsection{Example: Extension to exponentially decaying interaction terms}

As well as extending non-local operators using $H^*H$ outlined in the main text, we can extend the above rectangular truncations to non-local interactions. We consider a generalization of the lattice model on $l^2(\mathbb{Z})$ in the main text:
$$
[Hx]_n = \left(\frac{n^2}{10}+2i\sin(n)\right)x_n +\sum_{j=1}^\infty 2^{1-j}(x_{n-j}+x_{n+j}).
$$
Such non-local models are common in physical applications. To proceed, we note that if $\mathcal{P}_n$ denotes the projection onto $\mathrm{span}\{e_{-n},\ldots,e_n\}$ and $-n\leq j\leq n$, then
$$
\|(I-\mathcal{P}_{n+m})H\mathcal{P}_{n}x_j\|^2\leq 2 \sum_{s=m+1}^\infty 2^{2-2s}=\frac{8}{3}2^{-m}.
$$
Hence, we may choose $m$ (algebraically) dependent on $n$ such that
$$
\sigma_{\mathrm{inf}}(\mathcal{P}_{n+m}(H-zI)\mathcal{P}_n)\leq
\sigma_{\mathrm{inf}}((H-zI)\mathcal{P}_n)\leq
\sigma_{\mathrm{inf}}(\mathcal{P}_{n+m}(H-zI)\mathcal{P}_n)+2^{-n}.
$$
The above analysis now carries through by replacing $\gamma_N(z)$ with $\sigma_{\mathrm{inf}}(\mathcal{P}_{n+m}(H-zI)\mathcal{P}_n)+2^{-n}$. For example, \cref{fig:exponentially_decaying} shows verified pseudospectra and eigenvalues for this operator. Dealing with other long-range interactions is similar and boils down to controlling the tail of the error of rectangular truncations.

\begin{figure}[t]
\centering
\begin{minipage}{0.49\textwidth}
\includegraphics[width=1\textwidth,trim=0mm 0mm 0mm 0mm,clip]{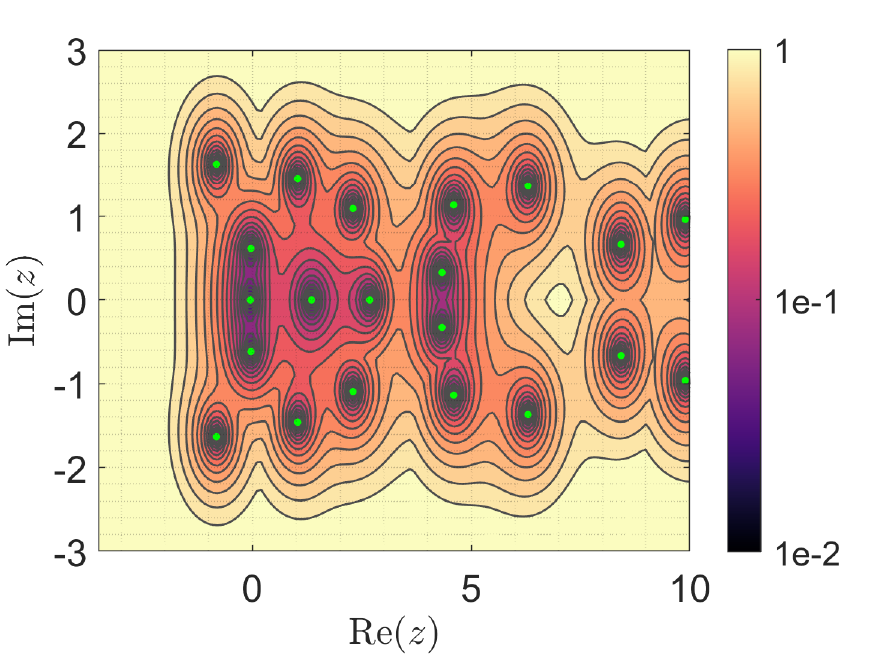}
\end{minipage}
\begin{minipage}{0.49\textwidth}
\begin{tabularx}{0.6\textwidth}{c}
Verified eigenvalue $\lambda_n$\\
\hline
$-0.04918293439$\hphantom{0000000000000000}\\
$-0.03617194872 \pm 0.61505608475i$\\
\hphantom{0}$1.35013464198$\hphantom{000000000000000}\\
\hphantom{00}$1.03403695407 \pm 1.45833018187i$\\
           $-0.82205220030 \pm 1.63118907210i$\\
\hphantom{00}$2.29590609739 \pm 1.09352704384i$\\
\hphantom{00}$2.67955625201$ \hphantom{000000000000000}\\
\hphantom{00}$4.33006751746 \pm 0.32817472558i$\\
\hphantom{00}$4.59741555217 \pm 1.13666442725i$\\
\hphantom{00}$6.29452856210 \pm 1.36850481036i$\\
\hphantom{00}$8.42195903570 \pm 0.66308931452i$\\
\hphantom{00}$9.90065477883 \pm 0.96114015191i$\\
\hline
\end{tabularx}
\end{minipage}
\caption{Illustration of the method for non-local model. Left: Verified pseudospectra. Right: A selection of the eigenvalues computed with verified absolute error $<10^{-10}$.}
\label{fig:exponentially_decaying}
\end{figure}



\end{document}